\pdfoutput=1
\RequirePackage{ifpdf}
\ifpdf 
\documentclass[pdftex]{sigma}
\else
\documentclass{sigma}
\fi

\numberwithin{equation}{section}

\newtheorem{Theorem}{Theorem}[section]
\newtheorem*{Theorem*}{Theorem}

\newtheorem*{Corollary*}{Corollary}

\newtheorem{Claim}[Theorem]{Claim}
\newtheorem{Proposition}[Theorem]{Proposition}
 { \theoremstyle{definition}
\newtheorem{Definition}[Theorem]{Definition}

\newtheorem{Construction}[Theorem]{Construction}
 }

\def \bla{{\boldsymbol{\lambda}}}
\def \bmu{{\boldsymbol{\mu}}}
\def \bth{{\boldsymbol{\theta}}}
\def \bdl{{\boldsymbol{\delta}}}
\def \bnu{{\boldsymbol{\nu}}}
\def \bet{{\boldsymbol{\eta}}}
\def \bpi{\boldsymbol\pi}
\def \bphi{\boldsymbol\phi}
\def \bro{\boldsymbol\rho}
\def \bLam{\boldsymbol\Lambda}

\def \si{\sigma}
\def \la{\lambda}
\def \be{\beta}
\def \al{\alpha}
\def \ta{\theta}
\def \dl{\delta}

\def \Dl{\Delta}
\def \CA{\mathcal A}
\def \CK{\mathcal M}
\def \CM{\mathcal M}
\def \CN{\mathcal N}
\def \CV{\mathcal V}

\def \CP{\mathcal P}

\def \CT{\mathcal T}
\def \BC{\mathbb{C}}
\def \BN{\mathbb{N}}

\def \BM{\mathbb{M}}

\def \BI{\mathbb{I}}
\def \jum{{\sf H}}

\begin{document}
\allowdisplaybreaks

\renewcommand{\thefootnote}{}

\renewcommand{\PaperNumber}{106}

\FirstPageHeading

\ShortArticleName{How to Draw a Correlation Function}

\ArticleName{How to Draw a Correlation Function\footnote{This paper is a~contribution to the Special Issue on Mathematics of Integrable Systems: Classical and Quantum in honor of Leon Takhtajan.

~~\,The full collection is available at \href{https://www.emis.de/journals/SIGMA/Takhtajan.html}{https://www.emis.de/journals/SIGMA/Takhtajan.html}}}

\Author{Nikolay BOGOLIUBOV and Cyril MALYSHEV}

\AuthorNameForHeading{N.~Bogoliubov and C.~Malyshev}

\Address{St.-Petersburg Department of Steklov Institute of Mathematics, RAS, \\ Fontanka 27, St.-Petersburg, Russia}
\Email{\href{mailto:email@address}{nmbogoliubov@gmail.com}, \href{mailto:malyshev@pdmi.ras.ru}{malyshev@pdmi.ras.ru}}

\ArticleDates{Received June 05, 2021, in final form December 02, 2021; Published online December 09, 2021}

\Abstract{We discuss connection between the $XX0$ Heisenberg spin chain and some aspects of enumerative combinatorics. The representation of the Bethe wave functions via the Schur functions allows to apply the theory of symmetric functions to the calculation of the correlation functions. We provide a combinatorial
derivation of the dynamical auto-correlation functions and visualise them in terms of nests of self-avoiding lattice paths. Asymptotics of the auto-correlation functions are obtained in the double scaling limit provided that the evolution parameter is large.}


\Classification{05A19; 05E05; 82B23}

\begin{flushright}
\begin{minipage}{53mm}
\it Dedicated to L.A.~Takhtajan\\
on occasion of his 70$\,{}^{th}$ birthday
\end{minipage}
\end{flushright}

\renewcommand{\thefootnote}{\arabic{footnote}}
\setcounter{footnote}{0}

\section{Introduction}

The quantum inverse scattering method turned out to be one of the most effective approaches to the solution of the quantum integrable interacting
many-body systems in low dimen\-sions~\cite{f, stf}. This method allows to obtain important results on the spin dynamics of the quantum
Heisenberg chain \cite{tf2, ft2, t, tf1}.

The zero anisotropy limit of the paradigmatic spin-$1/2$ $XXZ$ Heisenberg model is known as the $XX0$ chain. The correlation functions of the model are of considerable interest. Their behavior was intensively investigated for the system in the thermodynamic limit.
Connection between the $XX0$ chain and the low-energy QCD, as well as a possibility of third order phase transition \cite{gross} in the spin chain, are discussed in \cite{tier, zah, zah1, rzz}. The asymptotics of the partition functions, as well as the phase diagrams, for $XX0$ chain and related models are studied in~\cite{zah, zah1, rzz}.

The $XX0$ chain may be considered as a special case of free fermions \cite{col, its, lieb}.
Mathematical methods related with free fermions are useful in the study, for instance, of the Schur functions~\cite{buf, macd}, of random walks~\cite{ges}, of plane partitions \cite{bres} and three-dimensional Young diagrams~\cite{oko1}, of enumerative combinatorics \cite{stan1,stan2} and probability theory~\cite{borod, vers}.

Our paper is about the calculation of the dynamical auto-correlation functions of the $XX0$ chain \cite{nmb, bmnph, bmumn}. The representation of $N$-particles Bethe state-vectors in terms of the Schur functions and the well-developed theory of the symmetric functions allow us to express the transition amplitudes in the determinantal form. The introduction of
the off-shell Bethe state-vectors naturally connects the transition
amplitudes with the combinatorial objects. The vicious walkers were introduced in \cite{fish} and describe the situation in which two or more walkers arriving at the same lattice site annihilate each other. Two essentially
different types of the walkers paths are distinguished as lock steps and random turns models \cite{fish, forr1, 2}. Two main topologies of lock steps paths are known as stars and watermelons \cite{4, 5}. The graphical description of
the Schur functions is due to the bijection between semi-standard Young tableaux and stars \cite{4, 5}. The random turns paths naturally appear as
the transition amplitudes over the ferromagnetic states of the $XX0$ Hamiltonian. This approach allows us to interpret the analytical answers of
the auto-correlation functions as a superposition of the nests of self-avoiding directed lattice paths of vicious walkers.

The asymptotical estimates of the dynamical auto-correlation functions in the double scaling limit are found provided that the evolution parameter is large.
The amplitudes of the leading asymptotics are multiples of the numbers of watermelons or, equivalently, of the numbers of boxed plane partitions.

\textit{An outline of the content of this paper goes as follows}.
We start by the deriving
of the Cauchy--Binet type determinantal identities in
Section~\ref{sectdet}. The $q$-binomial determinants are
introduced and their connection with the generating functions of plane partitions is discussed.
The lattice paths configurations of star type are expressed in terms of the Schur functions in Section~\ref{comblatpat}. The generating functions of
the watermelon configurations are expressed in terms of the Cauchy--Binet type identities.
The norm-trace generating function of plane partitions in high box also follows
from the Cauchy--Binet type determinant, Section~\ref{comblatpat}.
The $XX0$ Heisenberg model on the cyclic chain is considered in
Section~\ref{randwalvic}. The Bethe eigen-vectors are expressed by means of the Schur functions, and the dynamical auto-correlation
functions are introduced.
The $N$-particles transition amplitudes over the ferromagnetic states are studied, and their interpretation in terms of the random turns walks is
given in Section~\ref{randfer}.
In Section~\ref{twopoint}, the ``two-time'' correlation function of projector over the ferromagnetic state is considered and interpreted in terms of the random turns paths as well. The correlation functions over
$N$-particles ground state are considered in Section~\ref{npart}. The
visualisation of these correlation functions is presented in terms of superposed random turns and lock steps walks. Section~\ref{asymp} deals with the asymptotics of the auto-correlation
functions. It is shown that the amplitudes of the asymptotical expressions are related with the generating functions of watermelons. We conclude with
a brief discussion in Section~\ref{conc}.

\section[The Schur functions and the Cauchy--Binet type determinantal identities]{The Schur functions and the Cauchy--Binet type\\ determinantal identities}\label{sectdet}

\subsection{Preliminary notations and definitions}

First, some definitions and conventions are in order. For example, boldface notations like ${\bf u}_N$, as well as ${\bf u}$, stand for $N$-tuples $(u_1, u_2, \dots , u_N)$ of $N$ (complex) numbers, and so on. In a real case, elements of a tuple are not necessarily ordered according to their values although might be either increasing or decreasing. We shall also use $N$-tuples like ${\bf M}_N= (M, M, \dots, M)$ or ${\bf k}_N= (k, k, \dots, k)$. The notation
$[N]\equiv \{1, 2, \dots , N\}$ is known and implies that the elements of the set are ordered.

Let a union of zero and of all natural numbers be denoted as $\bar\BN\equiv \{0 \}\cup {\BN}=
\{0, 1, 2, \dots \}$. An~$N$-tuple of strictly decreasing numbers $\mu_i \in \BN$, $1\leq i\leq N-1$, $\mu_N \in \bar\BN$
is called \textit{strict partition} ${\bmu}\equiv (\mu_1, \mu_2, \dots , \mu_N)$. The elements of $\bmu$ are called \textit{parts}, and they respect
\begin{gather}\label{calm}
M\geq\mu_1> \mu_2> \cdots > \mu_N \geq n .
\end{gather}
The \textit{length} of a partition, say, $\bmu$ is equal to the number of its parts, $l(\bmu)=N$. The \textit{weight} $|\bmu|$ of partition is equal to the sum of its parts, $|\bmu|=\sum_{i=1}^N \mu_i$.

An $N$-tuple of weakly decreasing non-negative integers provides another important partition ${\bla}\equiv (\la_1, \la_2, \dots, \la_N)$, where the parts $\lambda_j\in \bar\BN$ respect
\begin{gather}\label{ttt1}
{\cal L} \ge \la_1 \ge \la_2 \ge \dots \ge \la_N \ge n ,\qquad {\cal L}, n \in \bar\BN .
\end{gather}
The relationship
$\bla =\bmu - {\bdl}_N$, where
${\bdl}_N$ is the ``staircase'' partition
\begin{gather}
\label{stair}
{\bdl}_N \equiv (N-1, N-2, \dots, 1, 0) ,
\end{gather}
enables to connect the partitions $\bla$ and $\bmu$ so that ${\cal L}=M-N+1$ in \eqref{ttt1} since $\la_i=\mu_i-N+i$, $1\le i\le N$.

The partitions $\bla$ can be represented by Young diagrams consisting of $N$ rows of cells so that~$\la_i$ is the number of cells of $i^{\rm th}$ row, and cells are aligned to the left. A natural correspondence between the parts of $\bmu$ and $\bla$ is expressed by the Young diagram (see Figure~\ref{fig:f4}).

\begin{figure}[h]
\centering\includegraphics {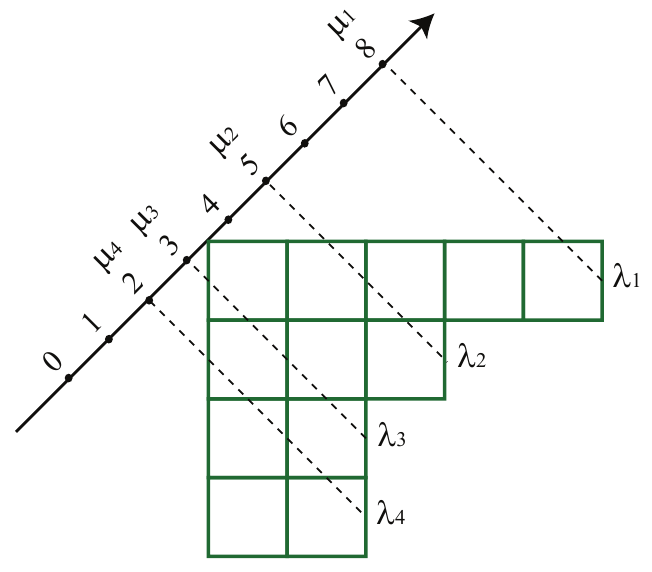}
\caption{Graphical correspondence between the Young diagram of partition $\bla=(5, 3, 2, 2)$ and the parts of strict partition $\bmu=(8, 5, 3, 2)$ for $M=8$, $N=4$.}
\label{fig:f4}
\end{figure}

Let us consider a partition $\bla_{N-k}$ of the length $l(\bla_{N-k})=N-k$, where $k\le N$, $N, k\in\bar\BN$. Proceeding with $\bla_{N-k}$, we shall use the notation $\widehat{\bla}\equiv \widehat{\bla}_N$ for
a partition of the length $l\big(\widehat{\bla}\big)=N$, which can be viewed as $\bla_{N-k}$ ``elongated'' by $k$ zero's as follows:
\begin{gather}\label{qanal14}
\widehat{\bla} \equiv (\la_1, \la_2, \dots , \la_{N-k}, 0, 0, \dots, 0) \equiv (\bla_{N-k}, 0, 0, \dots, 0) .
\end{gather}
It is appropriate to introduce a strict partition $\widehat{\bmu}$, $l(\widehat{\bmu})=N$:\vspace{-1ex}
\begin{gather} \label{strhat}
\widehat{\bmu}\equiv \widehat{\bla}+{\bdl}_N = ( {\bla}_{N-k} + {\bdl}_{N-k} +
{\bf k}_{N-k}, {\bdl}_{k}) =
(\bmu_{N-k} + {\bf k}_{N-k}, {\bdl}_{k}) ,
\end{gather}
where ${\bdl}_{N}$, ${\bdl}_{N-k}$, and ${\bdl}_{k}$ are the staircase partitions \eqref{stair} of the lengths $N$, $N-k$, and $k$, respectively, and ${\bf k}_{N-k}$ is a ``constant'' partition $(k, k, \dots , k)$ of the length $N-k$.

Let us introduce $k$-tuple ${\bf i} \equiv {\bf i}_k \equiv
(i_1, i_2, \dots, i_k)$
consisting of strictly increasing integers $1\le i_1 < i_2 < \cdots < i_k\le N$, $1\le k \le N$. It is appropriate to introduce a relative \textit {complement} of~${{\bf i}}$ in $[N]$ as $(N-k)$-tuple ${\complement_N {\bf i}}$:
\begin{gather}
\label{xindi2}
{\complement_N {\bf i}} \equiv
[N] \backslash {\bf i}
= \big(1, 2, \dots, \check{i}_1, \dots, \check{i}_2, \dots, \check{i}_k, \dots, N\big) ,
\end{gather}
where $\check{i}_l$ implies that the sequence $1, 2, \dots, N$
is missing the element ${i}_l$.

Let $N$-tuple of complex numbers ${\bf x}_N$ be given. Fixing $k$-tuple ${\bf i}$ and
its complement ${\complement_N {\bf i}}$ (\ref{xindi2}),
we introduce $k$-tuple ${\bf x}_{{\bf i}}$ and $(N-k)$-tuple ${\mathbf{x}_{_{\complement_N {\bf i}}}}$:
{\samepage\begin{gather}
\label{xindi}
\mathbf{x}_{{\bf i}} \equiv (x_{i_1}, x_{i_2}, \dots, x_{i_k}) ,\\
\label{xindi3}
{\mathbf{x}_{_{\complement_N {\bf i}}}}
\equiv (x_1, x_2, \dots, \check{x}_{i_1}, \dots,
\check{x}_{i_2}, \dots , \check{x}_{i_k}, \dots , x_N) \equiv
\overline{\mathbf{x}_{\bf i}} ,
\end{gather}}\noindent
where $\Check{x}_{i_l}$ implies that ${x}_{i_l}$ is dropped out of $N$-tuple $\mathbf{x}_{N}$.
The equivalent notation $\overline{\mathbf{x}_{\bf i}}
\equiv \mathbf{x}_{N} \backslash {\bf x}_{{\bf i}}$ (\ref{xindi3})
is to express that ${\mathbf{x}_{_{\complement_N {\bf i}}}}$ is also viewed as a relative complement of $k$-tuple ${\bf x}_{{\bf i}}$ (\ref{xindi}) in $\mathbf{x}_{N}$. Besides, we shall use the following notation:
\begin{gather}
\label{set1}
[N\backslash k]
\equiv [N] \backslash [k] .
\end{gather}
We shall also consider $\mathbf{x}_{[k]}$ and $\mathbf{x}_{[N\backslash k]}$,
as notations at ${\bf i} \equiv [k]$ for particular cases of ${\bf x}_{\bf i}$
and ${\bf x}_{[N]\backslash {\bf i}}$, respectively:
\begin{gather}
\label{xindi1}
\mathbf{x}_{[k]} \equiv (x_1, x_2, \dots, x_{k}) ,\qquad
{\mathbf{x}_{[N\backslash k]}}
\equiv (x_{k+1}, x_{k+2}, \dots, x_N)
\end{gather}
(it turns out that $\mathbf{x}_{[k]}$ and $\mathbf{x}_{k}$ denote the same).

Let us introduce a \textit{plane
partition of shape} $\bla$ as a map $\bpi$: $(i, j)
\rightarrow \pi_{i j}$, $(i, j)\in {\BN}^2$, from the Young diagram
of partition $\bla$ to $\bar\BN$ such that $\pi_{i j}$ is a non-increasing function of $i$ and $j$. The~entries $\pi_{i j}$
are called {\it parts} of the plane partition, and $|\bpi| =\sum_{i, j} \pi_{i j}$ is its {\it volume}.

\textit{Three-dimensional Young diagram} is a stack of unit cubes such that $\pi_{i j}$ is the height of the column with coordinates $(i,j)$. A \textit{box} ${\cal B}(L, N, K)$ of size $L\times N\times K$ is a subset of three-dimensional integer lattice:
\begin{gather*}
{\cal B} (L, N, K) \equiv \bigl\{(l, n, k)\in{\bar\BN}^3
\bigl| 0\leq l\leq L,\, 0\leq n\leq N,\,
0\leq k \leq K\bigr\} .
\end{gather*}
It is said that a plane partition $\bpi$ is contained in ${\cal B} (L, N, K)$ if
$i\leq L$, $j\leq N$, and $\pi_{i j} \leq K$ for all cubes of the Young diagram.

The \textit{generating function} $Z_q(L, N, K)$ of plane partitions $\bpi$ contained in ${\cal B}(L, N, K)$ is
of the form \cite{macd}:
\begin{gather}
Z_q(L, N, K) = \prod\limits_{j=1}^{L}
\prod\limits_{k=1}^{N}
\prod\limits_{i=1}^{K}
\frac{1-q^{i+j+k-1}}{1-q^{i+j+k-2}}
 = \prod\limits_{j=1}^{L}
\prod\limits_{k=1}^{N}
\frac{1-q^{K+j+k-1}}{1-q^{j+k-1}} .
\label{pf1}
\end{gather}
Right-hand side of (\ref{pf1}) gives at $q \to 1$ the number $A (L, N, K)$ of plane partitions in the box ${\cal B}(L,
N, K)$ (MacMahon’s formula):
\begin{gather}
A (L, N, K) = \prod\limits_{j=1}^{L}
\prod\limits_{k=1}^{N}
\prod\limits_{i=1}^{K}
\frac{i+j+k-1}{i+j+k-2}
 = \prod\limits_{j=1}^{L}
\prod\limits_{k=1}^{N}
\frac{K+j+k-1}{j+k-1} . \label{pf11}
\end{gather}

\subsection{The Cauchy--Binet type identities}

The {\it Schur functions}
$S_{\bla} ({\textbf x}_N)$ forming a basis for the ring of symmetric polynomials of $N$ variables are given by the relation (see \cite{fult, macd} for details):
\begin{gather}
S_{\bla} ({\textbf x}_N) \equiv
 S_{\bla} (x_1, x_2, \dots , x_N) \equiv \frac{\det\big(x_j^{\la_k+N-k}\big)_{1\leq
j, k \leq N}}{\CV({\textbf x}_N)},
\label{sch}
\end{gather}
where $\bla$ is a partition, and $\CV ({\textbf x}_N)$ is the Vandermonde determinant,
\begin{gather} \CV ({\textbf x}_N) \equiv
\det\big(x_j^{N-k}\big)_{1\leq j, k\leq N} =
\prod_{1 \leq m< l \leq N}(x_l-x_m) .
\label{spxx1}
\end{gather}

Bearing in mind the notations
(\ref{xindi2}), (\ref{xindi}), (\ref{xindi3}), (\ref{set1}), (\ref{xindi1}), we turn to

\begin{Proposition}\label{proposition1}
Let us choose $k$-tuple ${\bf i} = (i_1, i_2, \dots, i_k)$ and a partition ${\bla}_{N-k}$.
The relation for the Schur function $S_{{\bLam}} (\mathbf{x}_N)$ labelled by partition ${\bLam}$, $l({\bLam})=N$, holds true,
\begin{gather}\label{limprop}
\lim_{{\bf x}_{{\bf i}}\rightarrow 0} S_{{\bLam}} (\mathbf{x}_N) \equiv
\lim_{x_{i_1} \rightarrow 0} \lim_{x_{i_2} \rightarrow 0} \cdots \lim_{x_{i_k} \rightarrow 0} S_{{\bLam}} (\mathbf{x}_N) =
S_{{\bla}_{N-k}} ({\overline{ \mathbf{x}_{\bf i}}}) ,
\end{gather}
provided that $\bLam$ is of the form $\widehat{\bla}$
\eqref{qanal14}, and $\overline{\mathbf{x}_{\bf i}}$ \eqref{xindi3} is used in \eqref{limprop}.
\end{Proposition}

The representation (\ref{sch}) enables to prove (\ref{limprop}) provided that expansion of the determinants by minors is used iteratively when sending the elements of ${\bf x}_{{\bf i}}$ to zero.

The present section is concerned with the {\it Cauchy--Binet} type determinantal identity
for the Schur functions \cite{bmumn}: \begin{gather}
\mathcal{P}_{{\cal L} / n}(\textbf{x}_N, \textbf{y}_N)\equiv\sum_{\bla \subseteq
{\{({\cal L} \slash n)^N\}}}S_{\bla}(\textbf{x}_N)
S_{\bla}(\textbf{y}_N) =
\Bigg(\prod_{l=1}^N x_l^n y_l^n\Bigg) \frac{\det T(\textbf{x}_N, \textbf{y}_N) }{\CV (\textbf{x}_N)\CV (\textbf{y}_N)},
\label{scschr}
\end{gather}
where the summation goes over all
partitions $\bla$ of length $N$ with the parts satisfying \eqref{ttt1}.
The matrix $T(\textbf{x}_N, \textbf{y}_N)$ $\equiv$ $(T_{i j}(\textbf{x}_N, \textbf{y}_N))_{1\le i, j\le N}$ in (\ref{scschr}) is given by the
entries
\begin{gather}\label{tt}
T_{i j}(\textbf{x}_N, \textbf{y}_N)\equiv T_{i j}=\frac{1-(x_i y_j)^{N + {\cal L} - n}}{1-x_i y_j} .
\end{gather}
If we put, for shortness, $\mathcal{P}_{{\cal L}/0} (\textbf{x}_N, \textbf{y}_N) \equiv \mathcal{P}_{{\cal L}} (\textbf{x}_N, \textbf{y}_N)$, equation~\eqref{scschr} is re-expressed as
\begin{gather}\label{qanal7}
\mathcal{P}_{{\cal L}/n} (\textbf{x}_N,
\textbf{y}_N) = \mathcal{P}_{{\cal L}-n} (\textbf{x}_N,
\textbf{y}_N) \prod_{l=1}^N (x_l y_l)^n .
\end{gather}

Proposition~\ref{proposition1} allows us to go from the sum $\mathcal{P}_{\cal L} (\textbf{x}_{N}, \textbf{y}_N)$ (\ref{scschr}) to the sum $\mathcal{P}_{\cal L} (\overline{\mathbf{x}_{\bf i}},
\textbf{y}_N)$ one of the arguments of which is $(N-k)$-tuple $\overline{\mathbf{x}_{\bf i}}
\equiv \textbf{x}_{[N]\backslash {\bf i}}$ (\ref{xindi3}):
\begin{align}\label{qanal151}
\mathcal{P}_{\cal L} ( \overline{ \mathbf{x}_{\bf i}}, \textbf{y}_N) \equiv \sum\limits_{\bla \subseteq \{{\cal L}^{N-k}\}}
S_{\bla} (\overline{ \mathbf{x}_{\bf i}}) S_{\widehat\bla} ({\textbf y}_N) ,
\end{align}
where $k$-tuple ${\bf i}$ is fixed, summation is over $\bla$ of length $N-k$, and $\widehat{\bla}$ is given by (\ref{qanal14}). The determinantal identity for $\mathcal{P}_{\cal L} ( \overline{ \mathbf{x}_{\bf i}}, \textbf{y}_N)$, analogous to \eqref{scschr}, is given by

\begin{Theorem}\label{theoremI}
The Cauchy--Binet type identity is valid for $\mathcal{P}_{\cal L} (\overline{ \mathbf{x}_{\bf i}}, \textbf{y}_N)$
\eqref{qanal151}:
\begin{gather}
\mathcal{P}_{\cal L} (\overline{ \mathbf{x}_{\bf i}}, \textbf{y}_N) =
\frac{(-1)^{|{\bf i}| + kN- \frac{k}{2}(k-1)}}{\big(\prod_{l \in [N]\backslash {\bf i}}
x_{l}\big)^k} \,
\frac{\det {\overline T}(\overline{\mathbf{x}_{\bf i}}, \textbf{y}_{N})}{{\CV} (\overline{\mathbf{x}_{\bf i}}) {\CV} ({\textbf y}_N)} ,
\label{field43}
\end{gather}
where the entries of $N\times N$ matrix $\overline T(\overline{\mathbf{x}_{\bf i}}, \textbf{y}_{N}) \equiv (\overline T_{i j})_{1\le i, j\le
N}$ are
\begin{gather}
\nonumber
{\overline T}_{i j} = T^{\rm o}_{i j} ,\qquad i\in {\complement_N {\bf i}} = [N]\backslash {\bf i}, \qquad
1\le j \le N ,
\\
{\overline T}_{i j} = y_j^{k-l} ,\quad i \in \{i_l\}_{1\le l\le k} , \qquad
 1\le j \le N .
\label{field42}
\end{gather}
The entries $T^{\rm o}_{i j}$ are given by $T_{i j}$ \eqref{tt} taken at $n=0$.
\end{Theorem}

\begin{proof}
The overall summation in left-hand side of equation~(\ref{scschr}) taken at $n=0$
splits into two parts:
the sum over $\widehat\bla$ (\ref{qanal14}) (where ${\cal L} \geq \la_1 \geq \la_2 \geq\dots\geq \la_{N-k}\geq 0$) and the remaining sum.
We apply $\lim_{{\bf x}_{{\bf i}}\rightarrow 0}$ \eqref{limprop} in both sides of (\ref{scschr}). The relation $\mathcal{P}_{\cal L} ( \overline{ \mathbf{x}_{\bf i}}, \textbf{y}_N) = \lim_{{\bf x}_{{\bf i}}\rightarrow 0}
\mathcal{P}_{\cal L} (\textbf{x}_{N}, \textbf{y}_N)$ holds since
the only summation surviving in left-hand side of~(\ref{scschr}) is that
over $\widehat\bla$, and we thus obtain (\ref{field43}) due to Proposition~\ref{proposition1}.
\end{proof}

Two limits of the identity \eqref{field43} are of interest in what follows. In the first case, equation~(\ref{limprop}) is specified:
\begin{gather}
\label{limsch1}
\lim_{\mathbf{x}_{[k]} \rightarrow 0} S_{\widehat{\bla}} (\mathbf{x}_N) = S_{{\bla}_{N-k}}\big(\mathbf{x}_{[N\backslash k]}\big) ,
\end{gather}
where the relations
(\ref{xindi1}) are taken into account. In the second one, equation~(\ref{limprop}) reads:
\begin{gather}
\label{limsch}
\lim_{\mathbf{x}_{{\mathbf{N}_k} \backslash{{\bdl}_k}} \rightarrow 0} S_{\widehat{\bla}} (\mathbf{x}_N) = S_{{\bla}_{N-k}} (\mathbf{x}_{N-k}) ,
\end{gather}
where ${\mathbf x}_{\mathbf{N}_k \backslash {{\bdl}_k}} \equiv (x_{N-k+1}, x_{N-k+2}, \dots , x_{N-1}, x_N)$, and
the relation
\begin{gather}
\label{limsch2}
\mathbf{x}_{[N] \backslash ({\mathbf{N}_k} \backslash {{\bdl}_k})} = \mathbf{x}_{N-k}
\end{gather}
is accounted for.

Let us define a $q$-{\it
number} $[n]$ ($n\in\bar\BN$) and $q$-{\it factorial} $[n]! $ \cite{kac}:
\begin{gather}\label{qanal}
[n] \equiv \frac{1-q^n}{1-q} ,\qquad
[n]! \equiv [1] [2] \cdots [n] ,\qquad
[0]!\equiv 1 .
\end{gather}
This definition allows to define the {\it q-binomial coefficient}
$\left[\begin{smallmatrix}N\\r\end{smallmatrix}\right]$,
\begin{gather}\label{qanal1}
\begin{bmatrix}N\\r\end{bmatrix} \equiv
\frac{[N] [N-1] \cdots [N-r+1]}{[r]!} = \frac{[N]!}{[r]! [N-r]!} ,
\end{gather}
as well as the $q$-\textit{binomial determinant} \cite{car}:
\begin{gather}
\begin{pmatrix}
{\bf{ a}}\\
{\bf{ b}}
\end{pmatrix}_q \equiv \begin{pmatrix}
a_1, & a_2, & \cdots & a_S\\
b_1, & b_2, & \cdots & b_S
\end{pmatrix}_q \equiv
\det \begin{pmatrix}
\begin{bmatrix}
a_j\\b_i
\end{bmatrix}
\end{pmatrix}_{1\le i, j \le S},
\label{qanal2}
\end{gather}
where ${\bf{a}}$ and ${\bf{b}}$ are ordered $S$-tuples: $0\le a_1< a_2<\cdots<
a_S$ and $0\le b_1< b_2<\cdots< b_S$. In~the limit $q\to 1$,
the $q$-binomial determinant
(\ref{qanal2}) is transformed to the \textit{binomial determinant} since~(\ref{qanal1}) becomes the binomial coefficient
$\left(\begin{smallmatrix}N\\r\end{smallmatrix}\right)$ \cite{ges}.

The $q$-parametrization
\begin{gather}
\textbf{y}_N={\bf q}_N \equiv \big(q, q^2, \dots,q^N\big) ,\qquad
\textbf{x}_N={\bf q}_N/q = \big(1, q, \dots, q^{N-1}\big) \label{rep21}
\end{gather}
enables to represent the sum $\mathcal{P}_{\cal L} (\overline{ \mathbf{x}_{\bf i}}, \textbf{y}_N)$ (\ref{qanal151})
in the form:
\begin{gather}
\label{qanal4}
\mathcal{P}_{\cal L}
\bigg(\frac{\overline{ \mathbf{q}_{\bf i}}}{q}, {\bf q}_N \bigg)
= \sum\limits_{\bla \subseteq \{{\cal L}^{N-k}\}}
S_{\widehat\bla} ({\textbf q}_N)
S_{\bla} \bigg(\frac{ \overline{\mathbf{q}_{\bf i}}}{q} \bigg) ,
\end{gather}
where
\begin{gather*}
\overline{\mathbf{q}_{\bf i}} \equiv \big(q, q^2, \dots, \underset\smile {{ {q}^{i_{_1}}}}, \dots,
\underset\smile {{ {q}^{i_{_2}}}}, \dots , \underset\smile {{ {q}^{i_{_k}}}}, \dots , q^N\big) ,
\end{gather*}
and the underscored terms are missed. One arrives to the following

\begin{Claim} \label{Claim1} The
Cauchy--Binet type identity \eqref{field43} under the $q$-parametrization \eqref{rep21} is valid for $\mathcal{P}_{\cal L}
\big(\frac{\overline{ \mathbf{q}_{\bf i}}}{q}, {\bf q}_N \big)$ \eqref{qanal4}:
\begin{gather}
\mathcal{P}_{\cal L} \bigg(\frac{ \overline{ \mathbf{q}_{\bf i}}}{q}, {\bf q}_N \bigg)
= \frac{
(-1)^{|{\bf i}| + kN- \frac{k}{2}(k-1)}}
{q^{k (k-|{\bf i}|)+ \frac{k N}{2}(N-1)}}
\,
\frac{\det \overline {\sf T}
\bigl(\frac{\overline{ \mathbf{q}_{\bf i}}}{q}, {\bf q}_N \bigr)}{{\CV} \bigl(\frac{ \overline{\mathbf{q}_{\bf i}}}{q}\bigr) {\CV} ({\bf q}_N)} ,
\label{qanal3}
\end{gather}
where the matrix $\overline {\sf T} \bigl(\frac{\overline{ \mathbf{q}_{\bf i}}}{q}, {\bf q}_N \bigr)$ consists of $q$-parameterized
entries ${\overline{\sf
T}}_{i j}$ \eqref{field42}:
\begin{gather}
\nonumber
\overline{\sf T}_{i j} = \frac{[({\cal L}+N)(j+i-1)]}{[j+i-1]} ,\qquad
i\in [N]\backslash {\bf i}, \qquad 1\le j \le N ,
\\
\overline {\sf T}_{i j} = q^{j(k-l)} ,\qquad
 i \in \{i_l\}_{1\le l\le k} , \qquad 1\le j \le N ,
\label{qanal5}
\end{gather}
and $q$-numbers \eqref{qanal} are used.
\end{Claim}

The parametrisation of the identity (\ref{qanal3}) is concretized as ${\overline {\mathbf{q}_{\bf i}}} = { {\mathbf{q}_{[N\backslash k]}}}$ for ${\bf i}_k = [k]$
in the case \eqref{limsch1},
or ${\overline {\mathbf{q}_{\bf i}}} = {{\mathbf{q}_{N-k}}}$ for ${\bf i}_k = {\bf N}_k \backslash {\bdl}_k$ in the case \eqref{limsch}. The Schur functions $S_{\bla_{N-k}}
\bigl(
\frac{ \overline{ \mathbf{q}_{\bf i}}}{q} \bigr)$ in~these cases are related to each other (see (\ref{limsch2}) for the second case):
\begin{align}
\nonumber
S_{\bla_{N-k}}
\bigg(\frac{{\bf q}_{[N\backslash k]}}{q} \bigg) &\equiv
S_{\bla_{N-k}} (q^k, q^{k+1}, \dots, q^{N-1})
= q^{k|{\bla}|} S_{\bla_{N-k}}\bigg(\frac{{\bf q}_{N-k}}{q} \bigg)
\\
\label{schprop}
&= q^{k|{\bla}|} S_{\bla_{N-k}}
\bigg(\frac{{\bf q}_{[N] \backslash ({\mathbf{N}_k} \backslash {{\bdl}_k})}}{q} \bigg) .
\end{align}

{\sloppy
The Claim~\ref{Claim1} is to stress that equation~(\ref{qanal3}) generalizes the relationship between ${{\CP}}_{\cal L}
\bigl(\frac{{\bf q}_{N-k}}{q}, {\bf q}_N \bigr)$ and $\det \overline {\sf T}
\bigl(\frac{{\bf q}_{N-k}}{q}, {\bf q}_N \bigr)$ found in \cite{bmumn} and given by

\begin{Theorem} \label{theoremII}
The sum ${{\CP}}_{\cal L}
\bigl(\frac{{\bf q}_{N-k}}{q}, {\bf q}_N \bigr)$
expressed by \eqref{qanal4} at ${\bf i}_k = {\bf N}_k \backslash {\bdl}_k$ satisfies the
following identities:
\begin{align}
{\CP}_{\cal L}
\bigg(\frac{{\bf q}_{N-k}}{q}, {\bf q}_N \bigg) & = q^{-\frac k2(N-k-1)(N-k)} \frac{{\det} \overline{\sf T} \bigl(\frac{{\bf q}_{N-k}}{q}, {\bf q}_N\bigr)}{{\CV}({\bf
q}_{N}) {\CV}({\bf q}_{N-k}/q)} 
\label{qanal6}
\\[.3ex]
& = q^{-\frac N2({\cal L}-1) {\cal L}}
\begin{pmatrix}
2N-k, & 2N-k+1, & \cdots & 2N-k+{\cal L}-1\\
N-k, & N-k+1, & \cdots & N-k +{\cal L}-1
\end{pmatrix}_q \label{rep30}
\\
& = \prod_{l=1}^{{\cal L}}\prod_{j=1}^{N-k}\frac{[j+l+N-1]}{[j+l-1]} =Z_q (N-k, N, {\cal L}) .
\label{rep31}
\end{align}
\end{Theorem}}

\begin{proof}
Connection between the Schur functions and the elementary symmetric functions \cite{macd} allows one to derive (\ref{rep30}), \cite{bmumn}. Evaluation of the $q$-binomial determinant (\ref{rep30})
results in the double product representation (\ref{rep31}), which is the generating function of plane partitions in~${\cal B}(N-k, N, {\cal L})$ since $Z_q (N-k, {\cal L}, N)$ is equal to $Z_q (N-k, N, {\cal L})$), \cite{bmnph}.
\end{proof}

The $q$-version of (\ref{scschr}) arises from (\ref{rep30}), (\ref{rep31}) at $k=0$ due to (\ref{qanal7}):
\begin{align}\label{scqschr}
{{\CP}}_{{\cal L}/n} ( {\bf q}_{N}/q, {\bf q}_N) &= q^{nN^2} q^{\frac {N({\cal L}-n)}2 (1-{\cal L}+n)}
 \det \left(\displaystyle{
\begin{bmatrix} 2N+i-1 \\ N+j-1\end{bmatrix}}
\right)_{1\leq i, j\leq {\cal L}-n}\\
\label{scqschr12}
& = q^{nN^2} Z_q(N, N, {\cal L}-n) ,
\end{align}
where $Z_q (N, N, {\cal L}-n)$ is
the generating function (\ref{pf1}) of plane partitions in the box ${\cal B}(N, N,\allowbreak {\cal L}-n)$. The number of plane partitions $A (N, N, {\cal L}-n)$ is obtained as (\ref{pf11}):
\begin{gather}\label{npp}
A (N, N, {\cal L}-n)=\prod_{l=1}^{N} \prod_{j=1}^N
\frac{{\cal L}-n+j+l-1}{j+l-1} .
\end{gather}

\section{The Schur functions and watermelons}\label{comblatpat}

A combinatorial interpretation of the Cauchy--Binet type identities \eqref{scschr}, \eqref{field43}, and \eqref{qanal3}, is provided in the present section in terms of the watermelon configurations of self-avoiding lattice paths.

\subsection{The Schur functions and stars}

\subsubsection{Stars without deviation}

Let a set $\{\sf{T}\}$ of \textit{semi-standard Young tableaux ${\sf{T}}$ of shape} ${\bla}={\bla}_N$, \cite{fult}, with the entries taken from the set $[N]$ is given.
Each semi-standard tableau of shape ${\bla}$ may be represented as a \textit{nest of self-avoiding lattice paths} with prescribed starting and ending points. Consider a grid of vertical and horizontal dashed lines enumerated as in Figure \ref{fig:f5}.
\begin{figure}[h]
\centering
\includegraphics
{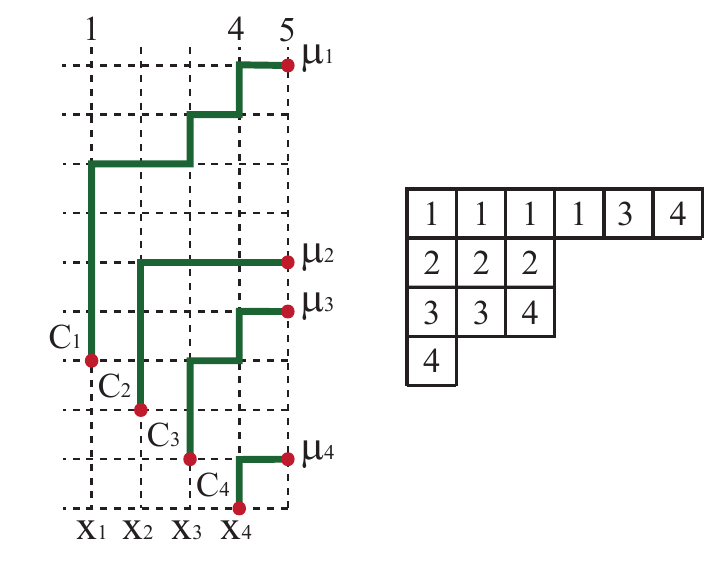}
\caption{The star $\cal{C}$ representing the semistandard tableau of shape $\bla=(6, 3, 3, 1)$.}
\label{fig:f5}
\end{figure}
The \textit{star} \cite{4, 5} is the nest of paths introduced by means of

\begin{Definition}\label{Definition3a}
The star $\cal{C}$, corresponding to the semi-standard Young tableau of shape ${\bla}$, is a~nest of $N$ self-avoiding lattice paths that connect (Figure~\ref{fig:f5}) the equidistantly arranged starting points $C_i=(i, N-i)$ with the non-equidistant ending points $(N, \mu_i)$, where
the parts of a strict partiton $\bmu$ respect \eqref{calm} at $n=0$, while the parts $\la_i=\mu_i-N+i$, $1\le i\le N$, respect \eqref{ttt1} with
${\cal L}={\cal M}$, where ${\cal M}\equiv M-N+1$.
\end{Definition}

The $i^{\rm th}$ row of the tableau ${\sf { T}}$ is encoded by $i^{\rm th}$ path of $\cal{C}$ making $\la_i$ steps upward. The~number of steps along the sites with abscissa $x_j$ is equal to the number of occurrences of $j$ in~the $i^{\rm th}$ row of ${\sf { T}}$. We put ${\bf x}_N\equiv (x_1, x_2, \dots, x_N)$ so that
the Schur function equivalent to \eqref{sch} is defined~\cite{fult}:
\begin{gather}\label{schrepr}
S_{\bla} ({\bf x}_N) = \sum_{\{{\sf { T}}\}} \prod_{i, j} x_{_{{\sf { T}}_{ij}}}= \sum_{\{\cal{C}\}}\prod_{j=1}^{N}
x_{j}^{c_j} ,
\end{gather}
where the middle sum is over all tableaux ${\sf { T}}$ of shape ${\bla}$, and the product of $x_{_{{\sf { T}}_{ij}}}$ is over all entries of ${\sf { T}}$. Summation in right-hand side of (\ref{schrepr}) is over all admissible $\cal{C}$, and $c_j$ is the number of steps upward along the vertical line with abscissa $x_j$. It follows from (\ref{schrepr}) that $S_{\bla} ({\bf 1}_N)$ is the total number of the nests of paths \cite{macd}:
\begin{gather}\label{numbpaths}
S_{\bla} ({\bf 1}_N) = \sum_{\{\cal{C}\}}1=\prod_{1\leq j < k \leq N}\frac{\lambda_j - j - \lambda_k+k}{k-j} .
\end{gather}

Each lattice path belonging to $\cal C$ can be enclosed by a rectangle so that the path's starting point coincides with the lower left vertex. The size of rectangle for $i^{\rm th}$ path is $\la_i\times (N-i)$, $1\le i \le N$. The {\it volume of the path} is the number of cells below the path within the corresponding rectangle. The volume $|{\cal C}|$ of the star $\cal C$ is equal to the sum of the volumes of the paths:
\begin{gather*}
|{\cal C}| = \sum_{j=1}^{N}(N-j)c_j =
N |{\bla}| - \sum_{j=1}^{N} j c_j ,
\end{gather*}
since $\sum_{j=1}^N c_j = |\bla|$ is the weight of ${\bla}$. The partition function of the stars $\cal C$ defined as ${\cal Z}_{\{\cal{C}\}} = \sum_{\{\cal{C}\}} q^{ \mid {\cal C} \mid}$ is expressed through the Schur function (\ref{schrepr}) at ${\bf x}_N = {\bf q}_N/q$:
\begin{gather}
{\cal Z}_{\{\cal{C}\}} = \sum_{\{\cal{C}\}} q^{ \mid {\cal C} \mid} = q^{N |{\bla}|} \sum_{\{\cal{C}\}}
q^{- \sum_{j=1}^{N} j c_j} = q^{N |{\bla}|}S_{\bla} \bigg(\frac{1}{\textbf{q}_N}\bigg)
= S_{\bla} \bigg(\frac{\textbf{q}_N}{q}\bigg)= q^{-|{\bla}|} S_{\bla} (\textbf{q}_N).
\label{schpar1}
\end{gather}
Equation \eqref{schpar1} expresses that $\{\cal{C}\}$
is determined by a choice of $\bla$.

Taking into account (\ref{schpar1}), it is appropriate to introduce the \textit{extended volume} $|{\cal C} |_{\rm w}$ (see Fi\-gure~\ref{fig:f5}):
\begin{gather}\label{dual10}
|{\cal C} |_{\rm w} \equiv
|{\cal C} | + |\bla|
= (N+1) |{\bla}| - \sum_{j=1}^{N} j c_j .
\end{gather}
Although the star in Figure~\ref{fig:f5} is presented for $N=4$, one {\it extra} column of cells is added between the lines with abscissae $N$ and $N+1$. To admit a single rightward step for each path in the nest is the same as to add $|\bla|$ to $|{\cal C}|$ thus obtaining $|{\cal C} |_{\rm w}$. The definition (\ref{dual10}) enables the following identity:
\begin{gather}
\sum_{\{\cal{C}\}} q^{ |{\cal C} |_{\rm w}} = q^{|{\bla}|} S_{\bla} \Bigl( \frac{\textbf{q}_N}{q} \Bigr) =S_{\bla} (\textbf{q}_N) .
\label{dual12}
\end{gather}

Let us introduce the partition ${\BM} \equiv ({\CK}, {\cal M}, \dots, {\CK})$ of the length $l({\BM})=N$, where ${\CK}$ is given by Definition~\ref{Definition3a}, and let us advance

\begin{Definition}
The conjugate star $\cal{B}$, corresponding to the semi-standard \textit{skew Young tableau of shape} ${\BM}\backslash {\bla}$, is a configuration of
$N$ self-avoiding lattice paths (as depicted in Figure \ref{fig:f6}) that connect the non-equidistant
points $(1, \mu_{\al})$, where $\mu_{\al}=\la_{\al}+N-{\al}$, $1\le {\al}\le N$, with the equidistantly arranged points $B_{\al}=({\al}, M+1-{\al})$.
\end{Definition}

\begin{figure}[h]\center
\includegraphics{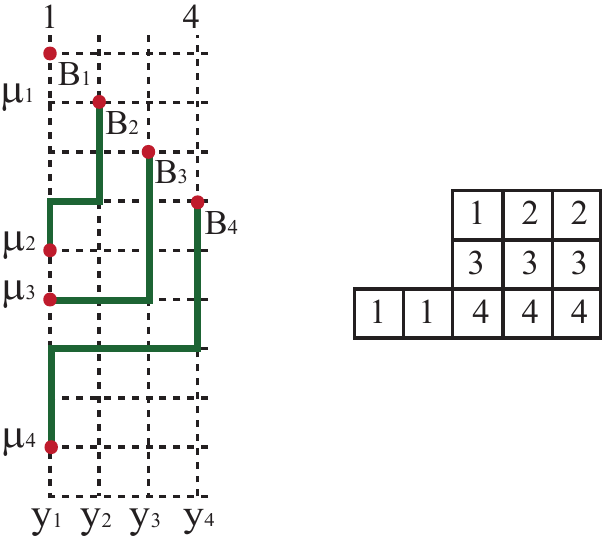}
\caption{Conjugate star ${\cal B}$ and the corresponding skew tableau
${\sf {T_{{\BM}\backslash {\bla}}}}$ of shape ${\BM}\backslash {\bla}$.}
\label{fig:f6}
\end{figure}

Each path in the conjugate star ${\cal B}$ is enclosed by rectangle of the size $({\CM}-\la_{\al})\times ({\al}-1)$. The volume of path in ${\cal B}$ is the number of cells \textit{above}
the path inside the corresponding rectangle (see Figure \ref{fig:f6}). The volume of ${\cal B}$ is $|{\cal B}|= \sum_{{\al}=1}^{N} b_{\al} ({\al}-1)$, where $b_{\al}$ is the number of steps along abscissa $y_{\al}$.

The $\al^{\rm th}$ path of $\cal{B}$ making ${\CM}-\la_{\al}$ upward steps
encodes the $\al^{\rm th}$ row of the skew tableau~${\sf { T_{{\BM}\backslash {\bla}}}}$ of shape ${\BM}\backslash {\bla}$. The number of steps along the sites with abscissa $y_{\be}$ is equal to the number of occurrences of $\be$ in $\al^{\rm th}$ row of ${\sf { T_{{\BM}\backslash {\bla}}}}$ (Figure~\ref{fig:f6}). Since ${\sf { T_{{\BM}\backslash {\bla}}}}$ and ${\sf { T_{ {\bla}}}}$ are in \textit{involution} \cite{step}, the corresponding Schur functions are coinciding. Thus we put ${\bf y}\equiv (y_1, y_2, \dots, y_N)$ and obtain
\begin{gather}\label{dual3}
S_{{\BM}\backslash \bla}({\bf y}) = \displaystyle{ \frac{\det\big(y_{\al}^{{\CM} -\la_{N-{\be}+1} +N-{\be}}\big)_{1 \leq {\al}, {\be} \leq N}}{\det(y_{\al}^{N-{\be}})_{1\leq
{\al}, {\be} \leq N}}} = \sum_{\{\cal{B}\}}\prod_{{\al}=1}^{N}
y_{{\al}}^{b_{\al}} .
\end{gather}

\begin{Proposition}\label{Proposition2}
\textit{The following representation of the Schur function $S_{\bla}({\bf y})$ is valid}:
\begin{gather}
\label{dual4}
S_{\bla} ({\bf y}) = \sum_{\{\cal{B}\}}\prod_{{\al}=1}^{N}
y_{{\al}}^{{\CM}-b_{\al}} .
\end{gather}
\end{Proposition}

\begin{proof}
Equation \eqref{dual4} is valid since factorizing out $\prod_{{\al}=1}^{N} y_{\al}^{\CM}$ from left-hand side of (\ref{dual3}), after some transformation one obtains \cite{stemb}:
\begin{gather*}
S_{{\BM}\backslash \bla}({\bf y}) = S_{\bla} \bigg(\frac{1}{\bf y}\bigg) \prod_{{\al}=1}^{N}
y_{{\al}}^{{\CM}} . \tag*{\qed}
\end{gather*}
\renewcommand{\qed}{}
\end{proof}

The representation (\ref{dual4}) connects the Schur function labelled by ${\bla}$ to the parameters of the conjugate stars ${\cal B}$.

Let us introduce two volumes of the conjugate star ${\cal B}$, namely, a dual volume $|{\cal B}|_{\rm d}$,
\begin{gather}\label{dual1}
|{\cal B}|_{\rm d} = \sum_{{\al}=1}^{N} ({\CM}-\la_{\al}-b_{\al})({\al}-1) ,
\end{gather}
and a dual extended volume $|{\cal B}|_{\rm w}$,
\begin{gather}\label{dual2}
|{\cal B}|_{\rm w} \equiv |{\cal B}|_{\rm d}+ n({\bla})=
\sum_{{\al}=1}^{N} ({\CM}-b_{\al})({\al}-1) ,\qquad n({\bla}) \equiv \sum_{{\al}=1}^{N} \la_{\al} ({\al}-1) .
\end{gather}
While $|{\cal B}|_{\rm d}$ (\ref{dual1}) ``counts'' the cells {\it below} the paths in ${\cal B}$, the definition (\ref{dual2}) gives $|{\cal B}|_{\rm w}$ as ``shifted''~$|{\cal B}|_{\rm d}$.

The partition function of the star $\cal B$, Figure~\ref{fig:f6},
is obtained from (\ref{dual4}) under the parametrization ${\bf y}={\bf q}_N/q$:
\begin{gather}\label{schpar2}
S_{\bla} \bigg(\frac{\textbf{q}_N}{q}\bigg) = \sum_{\{{\cal B}\}} q^{\sum_{j=1}^{N} (j-1) ({\CM} - b_j)} = \sum_{\{{\cal B}\}} q^{|{\cal B}|_{\rm w}} ,
\end{gather}
where equation~(\ref{dual2}) is used, and summation is over all stars $\cal{B}$.
It is seen from (\ref{dual3}) that $|{\cal B}|$ is connected with the Schur function that corresponds to the skew partition ${\BM} \backslash \bla$.
In the limit $q\rightarrow 1$, the Schur functions (\ref{dual12}) and (\ref{schpar2}) are equal to the number of stars either $\cal{B}$ or $\cal{C}$: $S_{\bla} ({\bf {1}}_N) = \sum_{\{\cal{B}\}} 1 = \sum_{\{\cal{C}\}} 1$.

\subsubsection{Stars with deviation}

Regarding the star ${\cal C}$ above, let us introduce the \textit{star ${\cal C}_k$ with deviation $k$} by means of

\begin{Definition}\label{Definition3b}
The star ${\cal C}_k$ with deviation $k$ is the nest of lattice paths like the nest ${\cal C}$ without upward steps along the lines with abscissae $x_1, x_2, \dots, x_k$.
The star ${\cal C}_k$ is characterized by such semi-standard
Young tableau of shape ${\bla}_{N-k}$ that the number inside the upper left corner cell is greater than $k$ (see Figure~\ref{fig:f555}).
\end{Definition}

The Schur function associated with the star ${\cal C}_k$ is represented like (\ref{schrepr}):
\begin{gather}
\label{schdevi}
S_{{\bla}} (\mathbf{x}_{[N\backslash k]}) = \sum_{\{{\cal C}_k\}}\prod_{j=k+1}^{N} x_{j}^{c_j} ,
\end{gather}
where ${\bla}\equiv {\bla}_{N-k}$ and $\mathbf{x}_{[N\backslash k]}$ (\ref{xindi1}) is used.
The representation (\ref{schdevi}) is in agreement with the representation (\ref{schrepr}) subjected to the limit (\ref{limsch1}) provided that $c_1=c_2=\dots = c_k = 0$.

\begin{figure}[h]\centering
\includegraphics
{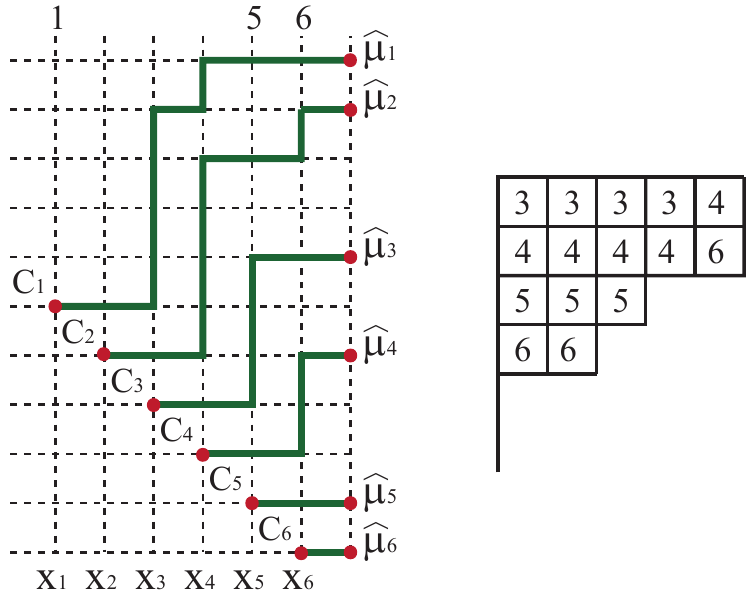}
\caption{The star ${\cal{C}}_k$ with deviation $k=2$ given by ${\widehat\bmu}=(10, 9, 6, 4, 1, 0)$ represents the semi-standard tableau of shape ${\widehat\bla}=(5, 5, 3, 2, 0, 0)$.}\label{fig:f555}
\end{figure}

By analogy with (\ref{dual10}), we introduce the extended volume $|{\cal C}_k |_{\rm w}$ of the star ${\cal C}_k$:
\begin{gather}
\label{dual100}
|{\cal C}_k |_{\rm w} \equiv
(N+1) |{\bla}_{N-k}| - \sum_{j=k+1}^{N} j c_j ,
\end{gather}
where $\sum_{j=k+1}^{N} c_j = |{\bla}_{N-k}|$ since $c_1=c_2=\dots = c_k= 0$, and the corresponding partition function is expressed:
\begin{gather}
{\cal Z}_{\{{\cal C}_k\}} =
\sum_{\{{\cal C}_k\}} q^{ \mid {\cal C}_k \mid_{\rm w}} =
q^{(N+1) |{\bla}_{N-k}|}
S_{{\bla}_{N-k}} \bigg(\frac{1}{\textbf{q}_{ [N\backslash k]}} \bigg) =
S_{{\bla}_{N-k}} ({ \textbf{q}_{N-k}}) .
\label{dual18}
\end{gather}
The modification of the volume in the form
\begin{gather}
\label{dual101}
|{\cal C}_k |_{\overline{\rm w}} \equiv
(N+k+1) |{\bla}_{N-k}| - \sum_{j=k+1}^{N} j c_j
\end{gather}
leads to the partition function
\begin{gather}
\label{dual181}
{\cal Z}_{\{{\cal C}_k\}} =
\sum_{\{{\cal C}_k\}} q^{ \mid {\cal C}_k \mid_{\overline{\rm w}}} = q^{k |{\bla}_{N-k}|}
S_{{\bla}_{N-k}} ({ \textbf{q}_{N-k}})
= S_{{\bla}_{N-k}} \big(\textbf{q}_{[N\backslash k]}\big) .
\end{gather}
The volumes \eqref{dual100} and \eqref{dual101} ensure \eqref{schprop}. Both (\ref{dual18}) and (\ref{dual181}) are reduced at $k=0$ to~(\ref{dual12}).

The conjugated star ${\cal B}_k$ with deviation $k$ is introduced by

\begin{Definition} \label{Definition4b}
The star ${\cal B}_k$ is the nest of lattice paths like ${\cal B}$ supplied with the requirement that the number of upward steps along a vertical line with abscissa $y_\alpha$ for $N-k+1 \leq \alpha \leq N$ is equal to ${\CM}$ (see Figure~\ref{fig:pathconj}).
\end{Definition}
\begin{figure}[h!]\centering
\includegraphics{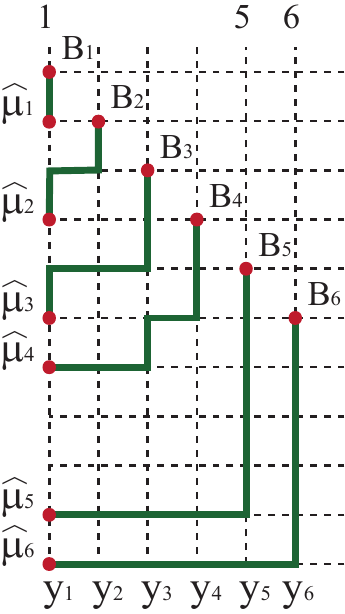}
\caption{The star representing skew tableau ${\BM} \backslash {\widehat{\bla}}$.}
\label{fig:pathconj}
\end{figure}

The limit (\ref{limsch}) results in the Schur function corresponding to a skew tableau ${\BM}\backslash {\widehat{\bla}}$, where~${\widehat{\bla}}$ is defined by (\ref{qanal14}). It is suggestive, regarding (\ref{dual4}), to put
\begin{gather}\label{lim}
S_{{\bla}_{N-k}} ({\bf y}_{N-k})=\sum_{\{{\cal B}_k\}} \prod\limits_{\al=1}^{N-k} y_{\al}^{{\CM}-b_{\al}} .
\end{gather}

Equation \eqref{lim} is expressed under the parametrization ${\bf y}={\bf q}_N/q$:
\begin{gather}
\nonumber
S_{{\bla}_{N-k}}\bigg(\frac{{\bf q}_{N-k}}{q}\bigg)
= \sum_{\{{\cal B}_k\}} q^{\mid {\cal B}_k \mid_{\rm w}} ,\qquad
| {\cal B}_k |_{\rm w} \equiv
\sum_{{\al}=1}^{N-k} ({\CM}-b_{\al})({\al}-1) .
\end{gather}

\subsection{Watermelons and their generating functions}\label{waterm}

\subsubsection{Watermelons without deviation}

The \textit{watermelon} configuration (see Figure~\ref{fig:f7}($a$)) is represented by the nest of self-avoiding lattice paths with equidistantly arranged starting $C_l$ and ending $B_l$ points ($1\le l \le N$). Only upward and rightward steps are allowed, and the nest is characterized by the paths with the total number of upward steps ${\CM}$ and the total number of rightward steps $N$. The $l^{\rm th}$ path in watermelon is contained within the rectangle such that $C_l$ and $B_l$ are its lower left and upper right vertices, respectively ($1\le l \le N$).

The nest of paths corresponding to watermelon results from

\begin{Construction}\label{Construction1}
Watermelon
is the nest of paths obtained by ``gluing'' the stars ${\cal C}$ and ${\cal B}$ along the dissection line determined by the partition $\bmu$ with parts respecting
\eqref{calm} so that the points $(N+1, \mu_i)$ (see Figure~\ref{fig:f5}) and
the points $(1, \mu_i)$ (see Figure~\ref{fig:f6})
are identified (Figure~\ref{fig:f7}).
\end{Construction}
\begin{figure}[h]
\center
\includegraphics {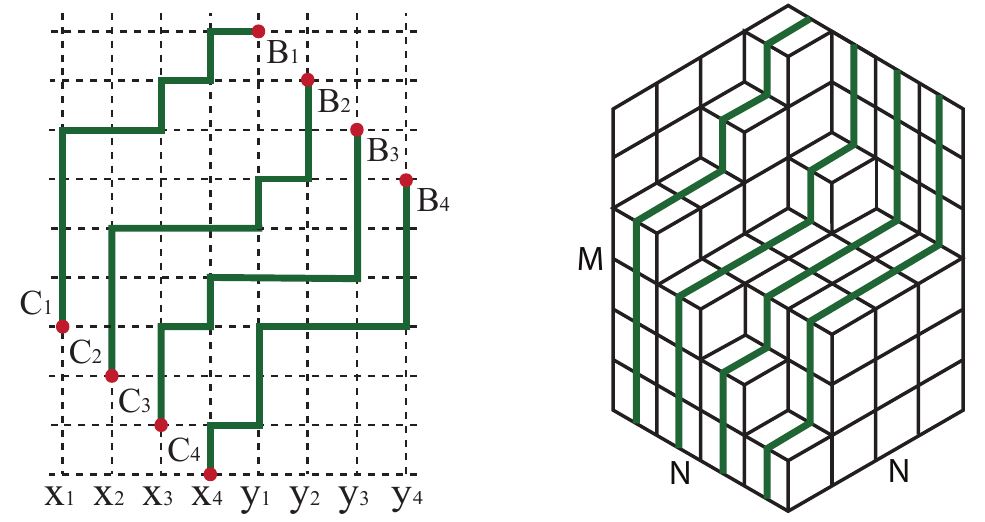}
\put(-220,-15){\makebox(0,0)[lb]{$a$}}
\put(-50,-15){\makebox(0,0)[lb]{$b$}}
\caption{Watermelon configuration of self-avoiding paths ($a$) and related boxed plane partition with \textit{gradient lines} ($b$).}
\label{fig:f7}
\end{figure}

{\sloppy
The representations of the Schur functions (\ref{schrepr}) and (\ref{dual4}) allow us to conclude that $\mathcal{P}_{{\cal L}/n}(\textbf{x}_N, \textbf{y}_N)$~(\ref{scschr}) is related with summation over the nests of paths constituting the water\-melon configurations, i.e.,
\begin{gather}
\label{gener1}
\mathcal{P}_{{\CM}/n} (\textbf{x}_N, \textbf{y}_N)
= \sum_{\bla\subseteq {\{({\CM}\slash n)^N\}}}
\Bigg(\sum_{\{\cal{C}\}} \prod_{j=1}^{N}x_{j}^{c_j}\Bigg)
\Biggl(\sum_{\{\cal{B}\}}\prod_{{\al}=1}^{N} y_{{\al}}^{{\CM}-b_{\al}} \Biggr) ,
\end{gather}
may be considered as the generating function of the watermelon configurations, which are in a~bijective correspondence with boxed plane partitions~\cite{4, 5}.

}

Let us define the \textit{volume of the path} as the number of cells below the path within the corresponding rectangle, and the \textit{volume of watermelon} as the volume of all paths constituting the watermelon (see Figure~\ref{fig:f7}).

\noindent
If so, the numbers of nests of the lattice paths constituting the watermelons are encoded by the \textit{generating function of watermelons} given by

\begin{Definition}\label{Definition5a}
Let ${\bf w}_{N, {\CM}\slash n}$ be a nest of $N$ paths constituting the watermelon. Each path in the nest consists of $N$ steps along the abscissa axis. The subscript ${\CM}/n$ implies that $\mu_N \ge n$ and the number of steps along the ordinate axis is ${\CM}$. The generating function $W_q (N, {\CM}/n)$ of watermelons is defined as the polynomial
\begin{gather}
W_q (N, {\CM}/ n)\equiv \sum_{\{{\bf w}_{N, {\CM}/ n}\}} q^{|{\bf w}_{N, {\CM}/n}|} ,
\label{gfw}
\end{gather}
where summation is over all nests ${{\bf w}_{N, {\CM}/n}}$, and $|{{\bf w}_{N, {\CM}/n}}|$ is the volume of the nest:
\begin{gather}
\label{gfw10}
|{{\bf w}_{N, {\CM}/n}}| \equiv\sum_{j=1}^{2N} (2N-j) m_j
- \frac{{\CM} N}{2} (N-1)+ \frac{nN}{2}(3N-1),
\end{gather}
and $m_j$ is the number of steps along the line with abscissa $x_j$, $1\le j\le 2 N$
(it is assumed at $n=0$ that ${{\bf w}_{N, {\CM}}}\equiv {{\bf w}_{N, {\CM}/0}}$).
\end{Definition}

Let us introduce the partition ${\widehat{\BM}}$ of length $l\big({\widehat{\BM}}\big) = N+L$ and volume $\big|{\widehat{\BM}}\big| = {\CM} N$,
\begin{gather}
\label{dual7}
{\widehat{\BM}} \equiv {\widehat{\BM}}_{N+L} = ({{\BM}}_N,\underbrace{0, 0, \dots, 0}_{L}) ,\qquad {\BM}_N \equiv (\underbrace{{\CM}, {\CM}, \dots, {\CM},}_{N}) .
\end{gather}
The corresponding semi-standard Young tableau of shape ${\widehat{\BM}}$ consists of cells arranged in $N$ rows of length $N$ (and $L$ ``rows'' of zero length). Enumeration of the sets of numbers which ``fill'' the cells is equivalent to enumeration of $2N$-tuples
$(m_1, m_2, \dots, m_{2N})$.
With regard at the watermelon, Figure~\ref{fig:f7}, we put $N=L$ and define the corresponding Schur function $S_{\widehat{\BM}}$ labelled by ${\widehat{\BM}}$ \eqref{dual7}:
\begin{gather}\label{schwat}
S_{\widehat{{\BM}}} ({\bf x}_{2N}) \equiv \sum_{\{{\bf w}_{N, {\CM}}\}}\prod_{j \in [2N]} x_j^{m_j} .
\end{gather}

Taking into account \eqref{schwat}, we arrive at the following

\begin{Proposition}\label{Proposition3}
The generating function $W_q (N, {\CM}/n)$ \eqref{gfw} satisfies the identities:
\begin{align}
\label{volwater3}
W_q (N, {\CM}/n) & = \mathcal{P}_{{\cal M}/n}\bigg(\textbf{q}_N, \frac{\textbf{q}_N}{q}\bigg)\\
\label{gfw5}
& = q^{nN(N-1) - ({\CM}-n) \frac{N}{2}(N+1)} S_{\widehat{\BM}}
(\textbf{q}_{2N}) ,
\end{align}
where $\mathcal{P}_{{\cal M}/n}\bigl(\textbf{q}_N, \textbf{q}_N/q\bigr)$ is the polynomial \eqref{gener1} under $q$-parametrization \eqref{rep21}, and $S_{\widehat{\BM}}
(\textbf{q}_{2N})$ is given by \eqref{schwat} under the $q$-parametrization $\textbf{x}_{2N} = \textbf{q}_{2N}$.
\end{Proposition}

\begin{proof}
{\sloppy
First, let us note that $|{\bf w}_{N {\CM}}|$ defined by (\ref{gfw10}) at $n=0$ satisfies the identity
\begin{gather}
\label{volwater}
|{\cal C}|_{\rm w}^{{\CM}} + |{\cal B}|_{\rm w}^{{\CM}} = |{\bf w}_{N {\CM}}| ,
\end{gather}
where the notations $|{\cal C}|_{\rm w}^{{\CM}}$ and $|{\cal B}|_{\rm w}^{{\CM}}$ are introduced instead of the volumes $|{\cal C}|_{\rm w}$ (\ref{dual10}) and~$|{\cal B}|_{\rm w}$~(\ref{dual2}), respectively, for
notational convenience to remind that $\bla$ satisfies \eqref{ttt1}
(${\cal L}={\CM}$, $n=0$), whereas the ``gluing'' partition $\bmu$ satisfies \eqref{calm} ($n=0$).

}

We also introduce the notations $|{\cal C}|_{\rm w}^{{\CM}/n}$ and $|{\cal B}|_{\rm w}^{{\CM}/n}$ for the extended volumes of the nests of paths characterized by $\bmu$ \eqref{calm} at arbitrary $n$:
\begin{align}\label{volwater4}
& |{\cal C}|_{\rm w}^{{\CM}/n} \equiv
|{\cal C}|_{\rm w}^{{\CM}-n} + \frac{nN}{2}(N+1),
\\
& |{\cal B}|_{\rm w}^{{\CM}/n} \equiv
|{\cal B}|_{\rm w}^{{\CM}-n} + \frac{nN}{2}(N-1).
\label{volwater5}
\end{align}
The volumes $|{\cal C}|_{\rm w}^{{\CM}/n}$ \eqref{volwater4} and $|{\cal B}|_{\rm w}^{{\CM}/n}$ \eqref{volwater5} satisfy the identity
\begin{gather*}
|{\cal C}|_{\rm w}^{{\CM}/n} + |{\cal B}|_{\rm w}^{{\CM}/n} = |{\bf w}_{N, {\CM} - n}| + nN^2 = |{\bf w}_{N, {\CM}/n}| ,
\end{gather*}
where \eqref{volwater} is taken into account.

Left-hand side of \eqref{volwater3} is re-expressed with the use of \eqref{volwater4} and \eqref{volwater5}:
\begin{gather}
W_q (N, {\CM}/n) = \sum_{\bla \subseteq
\{({\cal M}-n)^N \}}
\bigg(q^{\frac{nN}{2}(N+1)} \sum_{\{{\cal C}\}} q^{|{\cal C}|_{\rm w}^{{\cal M}-n}}\biggr)
\biggl(q^{\frac{nN}{2}(N-1)} \sum_{\{{\cal B}\}} q^{|{\cal B}|_{\rm w}^{{\cal M}-n}}\biggr) , \label{volwater1}
\end{gather}
where the summation over the nests ${\bf w}_{N {\cal M}}$ is replaced equivalently by the summation over the stars ${\cal C}$ and ${\cal B}$:
\begin{gather}\label{graph}
\sum_{\{{\bf w}_{N {\cal M}}\}} \Longleftrightarrow \sum_{\bla \subseteq
\{{\cal M}^N \}}
\sum_{\{{\cal C}\}} \sum_{\{{\cal B}\}} .
\end{gather}
Further, we use the representations (\ref{dual12}), (\ref{schpar2}), and
obtain from \eqref{volwater1}:
\begin{gather}
W_q (N, {\CM}/n) = q^{n N^2} \mathcal{P}_{{\cal M}-n}\biggl(\textbf{q}_N, \frac{\textbf{q}_N}{q}\biggr) = \mathcal{P}_{{\cal M}/n}\biggl(\textbf{q}_N, \frac{\textbf{q}_N}{q}\biggr),
\label{volwater2}
\end{gather}
where the $q$-version of the Cauchy--Binet identity (\ref{scqschr}) is accounted for. The identity \eqref{volwater3} is justified.

Furthermore, we obtain from (\ref{gfw}) and (\ref{gfw10}):
\begin{gather*}
W_q (N, {\CM}/n) = q^{(2{\CM}+n)N^2 - \frac{({\CM}-n) N}{2}(N-1)}
S_{\widehat{\BM}}
\bigg(\frac{1}{\textbf{q}_{2N}} \bigg) ,
\end{gather*}
where $S_{\widehat{\BM}}
(1/\textbf{q}_{2N})$ is given by
\eqref{schwat} at $\textbf{x}_N = 1/\textbf{q}_{2N}$.
Taking into account \eqref{schpar1} we obtain:
\begin{gather}
\nonumber
S_{\widehat{\BM}}\bigg(\frac{1}{\textbf{q}_{2N}} \bigg) = q ^{- 2{\CM} N^2}
S_{\widehat{\BM}}
\bigg(\frac{\textbf{q}_{2N}}{q} \bigg) = q ^{- 2{\CM} N^2 - {\CM} N}
S_{\widehat{\BM}}({\textbf{q}_{2N}}) ,
\end{gather}
and therefore \eqref{gfw5} is also valid.
\end{proof}

\subsubsection{Watermelons with deviation}

To proceed further, \textit{watermelon with deviation} is introduced by means of

\begin{Construction}\label{Construction2}
Watermelon with deviation $k$ (see Figure \ref{fig:f10}) is obtained by ``gluing'' the stars~${\cal C}_k$ and ${\cal B}$ (see Figure \ref{fig:f5}), along the dissection line determined by
$\widehat{\bmu}$ (\ref{strhat}).
\end{Construction}
\begin{figure}[h]
\center
\includegraphics{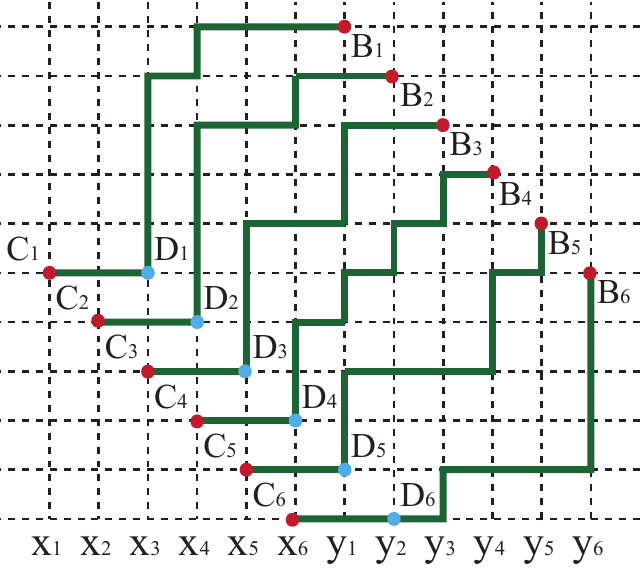}
\caption{Watermelon with deviation $k=2$ and $N=6$.}
\label{fig:f10}
\end{figure}

It follows from Construction~\ref{Construction2} that
the generating function of watermelons with devia\-tion~$k$
is given by the sum:
\begin{align}
\nonumber
\mathcal{P}_{\CK} ({\textbf x}_{N;k}, \textbf{y}_N) & \equiv \sum \limits_{\bla \subseteq \{{\CK}^{N-k}\}} S_{\bla} ({\textbf x}_{N;k})
S_{\widehat\bla} ({\textbf y}_N)
\\
\label{gendev}
&= \sum\limits_{\bla \subseteq \{{\CK}^{N-k}\}}
\Biggl(\sum_{\{{\cal C}_k\}}\prod_{x_{j}\in {\bf x}_{N;k}}x_{j}^{c_j}\Biggr)
\Biggl( \sum_{\{\cal{B}\}} \prod_{{\al}=1}^{N}y_{{\al}}^{{\CM}-b_{\al}} \Biggr),
\end{align}
where ${\bf x}_{N;k}$ implies either ${\bf x}_{N-k}$ or ${\bf x}_{[N\backslash k]}$.
Another representation of the watermelon with deviation may be obtained by gluing the stars ${\cal C}$ and ${\cal B}_k$ (see Figure~\ref{fig:wm}), and the generating function reads:
\begin{align}
\mathcal{P}_{\CK} (\textbf{x}_{N},\textbf{y}_{N-k}) & \equiv \sum\limits_{\bla \subseteq \{{\CK}^{N-k}\}}S_{\widehat\bla} ({\textbf x}_{N})
S_{\bla} ({\textbf y}_{N-k}) \nonumber
\\
&= \sum\limits_{\bla \subseteq \{{\CK}^{N-k}\}}\Biggl(\sum_{\{{\cal C}\}} \prod_{j=1}^{N}x_{j}^{c_j}\Biggr)
\Biggl(\sum_{\{{\cal{B}}_{k}\}}\prod_{{\al}=1}^{N-k}y_{{\al}}^{{\CM}-b_{\al}} \Biggr) .
\label{gendev2}
\end{align}
\begin{figure}[h]
\center
\includegraphics{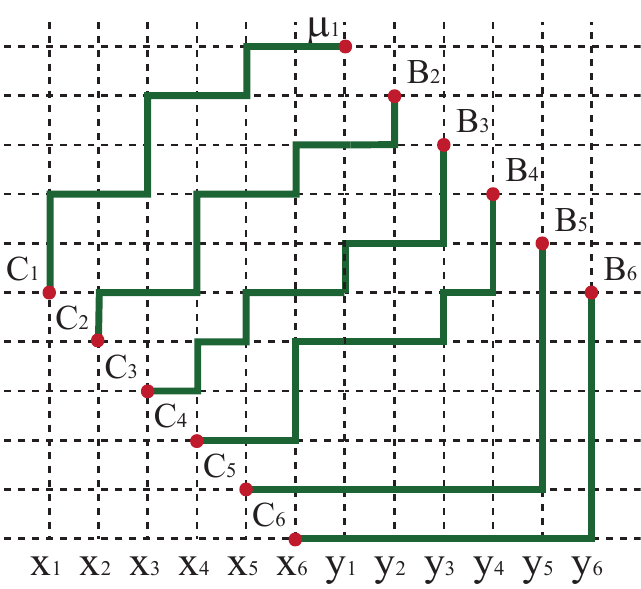}
\caption{Conjugated watermelon with deviation $k=2$ and $N=6$.}
\label{fig:wm}
\end{figure}

We introduce the generating function ${W}_q (N, L, {\CM})$ of the watermelons with deviation $k$ by means of

\begin{Definition}\label{Definition5b}
Let a nest ${{\bf w}_{N L {\CM}}}$ of $N$ lattice paths characterized by
the total numbers $L$ and ${\CM}$ of steps along abscissa and
ordinate axes to constitute the watermelon with deviation $N-L$.
The generating function $W_q (N, L, {\CM})$ of the nest ${{\bf w}_{N L {\CM}}}$ is given by the polynomial
\begin{gather}
W_q (N, L, {\CM}) \equiv \sum_{\{{\bf w}_{N L {\CM}}\}} q^{|{\bf w}_{N L {\CM}}|_{\dl}} ,
\label{gfw1}
\end{gather}
where summation goes over all admissible ${{\bf w}_{N L {\CM}}}$. Let $\dl$ to specify the volume of ${\cal C}_k$ used in Construction~\ref{Construction2}: the choice $\dl= 0$ or $\dl=k$ corresponds to the volume either \eqref{dual100} or \eqref{dual101}. The
corresponding volumes $|{{\bf w}_{N L {\CM}}}|_{\dl}$ are parameterized by $\dl$:
\begin{gather}
\label{wmdvol}
|{\bf w}_{N L {\CM}} |_{\dl} \equiv \sum_{j=k+1}^{2 N}(2N-j) m_j + \dl \sum_{j=k+1}^{N} m_j
-\frac{{\CM} N}{2} (N-1) .
\end{gather}
The numbers of steps along the vertical lines with abscissae $x_j$, $k+1\le j \le 2 N$, respect $\sum_{j=k+1}^{2 N} m_j ={\CM} N$ since \eqref{wmdvol} is right-hand side of \eqref{gfw10} at $m_1=m_2= \dots = m_k =0$ and $k=N-L$.
\end{Definition}

\vskip0.3cm
In the case $N- L=k$, we define the Schur function labelled by $\widehat{\BM}$ \eqref{dual7} in the form similar to \eqref{schwat}:
\begin{gather}
\label{schwatt}
S_{\widehat{{\BM}}}({\textbf x}_{[2N\backslash k]}) \equiv \sum_{\{{\bf w}_{NL{\CM}}\}}\prod_{j=k+1}^{2N} x_j^{m_j} .
\end{gather}
Then the graphical considerations enable to formulate

\begin{Proposition}\label{Proposition4}
The generating function $W_q (N, L, {\CM})$ \eqref{gfw1} of watermelons with deviation $N-L=k$ satisfies the identities:
\begin{align}
\label{volwaterk3}
W_q (N, L, {\CM}) & = \mathcal{P}_{{\cal M}} \bigg(\textbf{q}_{{N; k}}, \frac{\textbf{q}_N}{q}\bigg)
\\
\label{volwaterk31}
&= q^{-\frac{{\CM}N}{2}(N+1)}
S_{\widehat{\BM}} \bigl(\textbf{q}_N, q^{\dl+N+1}, q^{\dl+N+2},
\dots, q^{\dl+2N-k}\bigr) ,
\end{align}
where $\mathcal{P}_{{\cal M}}\bigl(\textbf{q}_{N; k}, \textbf{q}_N/q\bigr)$ and the Schur function $S_{\widehat{\BM}}$ are
given by \eqref{gendev} and \eqref{schwatt}, respectively, under the $q$-parametrization \eqref{rep21}. The notation $\textbf{q}_{{N; k}}$ implies $\textbf{q}_{N - k}$ at $\dl=0$
or $\textbf{q}_{{N; k}} = \textbf{q}_{[N\backslash k]}$ at $\dl= k$.
\end{Proposition}

\begin{proof} The volume $|{\bf w}_{N L {\CM}}|_{\dl}$ (\ref{wmdvol}) respects the relationship
\begin{gather}
\label{volwaterk}
|{\cal C}_k|_{\rm w}^{\dl} + |{\cal B}|_{\rm w} = |{\bf w}_{NL{\CM}}|_{\dl} ,
\end{gather}
where $|{\cal B}|_{\rm w}$ is the volume (\ref{dual2}), and the superscript $\dl$ in $|{\cal C}_k|_{\rm w}^{\dl}$ is to stress that either $|{\cal C}_k|_{\rm w}$ \eqref{dual100} or $|{\cal C}_k|_{\overline{\rm w}}$ \eqref{dual101} is used at $\dl=0$ or $\dl=k$, respectively.
Equations \eqref{gfw1} and \eqref{volwaterk} lead us to the following relation:
\begin{gather}
W_q (N, L, {\CM}) = \sum_{\bla \subseteq\{{\cal M}^L \}}
\Biggl(\sum_{\{{\cal C}_k\}} q^{|{\cal C}_k|_{\rm w}^{\dl}}\Biggr)
\Biggl(\sum_{\{{\cal B}\}} q^{|{\cal B}|_{\rm w}}\Biggr) , \label{volwaterk1}
\end{gather}
since the summation $\sum_{\bla \subseteq
\{{\cal M}^L \}}
\sum_{\{{\cal C}_k\}} \sum_{\{{\cal B}\}}$ replaces the sum $\sum_{\{{\bf w}_{N L {\CM}}\}}$ (compare with (\ref{graph})). Since the product
of (\ref{schpar2}) and (\ref{dual18}) is expressed as the product of the $q$-parametrized Schur functions, right-hand side of (\ref{volwaterk1})
is re-expressed:
\begin{gather}
\sum_{\bla \subseteq
\{{\cal M}^{L} \}} S_{\bla} (\textbf{q}_{N; k})
S_{\widehat\bla} \bigg(\frac{\textbf{q}_N}{q}\bigg)
= \mathcal{P}_{{\cal M}} \bigg(\textbf{q}_{N; k}, \frac{\textbf{q}_N}{q}\bigg) .
\label{volwaterk2}
\end{gather}
Equation \eqref{volwaterk3} is valid due to \eqref{volwaterk1} and \eqref{volwaterk2}.

Let us turn to equation~\eqref{volwaterk31}. We obtain from \eqref{gfw1}, \eqref{wmdvol} and \eqref{schwatt}:
\begin{gather}
\label{volwaterk303}
W_q (N, L, {\CM}) = q^{-\frac{N{\CM}}{2}(N-1) + 2 N^2 {\CM}}
S_{\widehat{\BM}} \bigg(\frac{q^{2\dl-k}}{{\bf q}_{N; k}},
\frac{1}{{\bf q}_{[2N\backslash N]}}\bigg) .
\end{gather}
We re-express $W_q(N, L, {\CM})$ (\ref{volwaterk303}) at $\dl=0$:
\begin{align}
\nonumber
W_q (N, L, {\CM}) & = q^{-\frac{N{\CM}}{2}(N-1) + 2 N^2 {\CM}}
S_{\widehat{\BM}} \bigg(\frac{1}{q^{k+1}},
\frac{1}{q^{k+2}}, \dots,\frac{1}{q^{2N}}\bigg)
\\
\nonumber
& = q^{-\frac{N{\CM}}{2}(N-1) + (2 N - k) {\CM}N}
S_{\widehat{\BM}} \bigg(\frac{1}{{\bf q}_{2N-k}}\bigg)
\\
\label{wmdpf3}
& = q^{-\frac{{\CM}N}{2}(N+1)}S_{\widehat{\BM}}({{\bf q}_{2N-k}}) .
\end{align}
Further, we obtain at $\dl=k$:
\begin{gather}
\nonumber
S_{\widehat{\BM}} \bigg(\frac{1}{q},
\frac{1}{q^{2}},\dots, \frac{1}{q^{N-k}}, \frac{1}{q^{N+1}}, \dots,
\frac{1}{q^{2N}}\bigg)
\\ \qquad
{} = q^{- {\CM}N (2 N+1) }
S_{\widehat{\BM}} \bigr(q^{2N}, q^{2N-1}, \dots, q^{N+k+1}, q^{N}, q^{N-1}, \dots, q^2, q \bigl) .
\label{wmdpf33}
\end{gather}
Equation \eqref{volwaterk31} is thus valid due to \eqref{volwaterk303} and \eqref{wmdpf33}.
\end{proof}

\begin{Corollary*}
The Schur function $S_{\widehat{\BM}}$ \eqref{volwaterk31} acquires
the determinantal representation due to Proposition~$\ref{Proposition4}$ provided that \eqref{qanal4}, \eqref{qanal3}, \eqref{qanal5} are taken into account:
\begin{gather}\nonumber
 S_{\widehat{\BM}} \bigl(\textbf{q}_N, q^{\dl+N+1}, q^{\dl+N+2},\dots, q^{\dl+N+L}\bigr)
\\ \qquad
\label{werk12}
{}= q^{\frac{{\CM}N}{2}(N+1) - \frac{L}{2} (N - L)(2\dl + L-1)} \frac{\det {\overline T}^{\dl}} {{\CV} ({\bf q}_N/q) {\CV} (\mathbf{q}_{N; k})},
\end{gather}
where the entries of $N\times N$ matrix ${\overline T}^{\dl}$ are given:
\begin{gather}
{\overline T}_{i j}^{\dl} \equiv \frac{[({\CM}+N) (\dl+j+i-1)]}{[\dl+j+i-1]},\qquad
1\le i\le N-k, \qquad 1\le j \le N ,\nonumber
\\
{\overline T}_{i j}^{\dl} \equiv q^{j(N-i)} ,\qquad
 N-k+1\le i\le N , \qquad 1\le j \le N .\label{qanal55}
\end{gather}
\end{Corollary*}

Provided that $k=0$,
Proposition~\ref{Proposition4} is reduced to Proposition~\ref{Proposition3} taken at $n=0$.

\subsubsection[The generating function Wq (N, L, M) as determinant]
{The generating function $\boldsymbol{W_q (N, L, {\CM})}$ as determinant}

The generating function $W_q (N, {\CM})$ \eqref{gfw}
expressed by \eqref{volwater3}
acquires the $q$-binomial determinant form due to (\ref{scqschr}). On the other hand, equation~(\ref{gfw5}) represents $W_q (N, {\CM}/n)$ as the $q$-parameterized Schur function. Although the Schur function representation \eqref{volwaterk31} for $W_q (N, L, {\CM})$ depends on $\dl$, the corresponding limits at $q\to 1$ coincide for $\dl=0$ and $\dl=k$.

The determinantal form of $W_q (N, L, {\CM})$ at $\dl=0$ is the subject of

\begin{Theorem}\label{TheoremIII}
The partition function ${W}_q (N, L, {\CM})$ of the watermelon with deviation $N-L\allowbreak=k$ respects the determinantal representation:
\begin{align}
\label{qdetgv1}
{W}_q(N, L, {\CM}) & = q^{-\frac{{\CM} N}{2} (N-1)} \det \left(q^{(j-1)({\CM}+j-i)}
\begin{bmatrix}L+{\CM}+N-i \\ {\CM}+N-j \end{bmatrix}
\right)_{1\leq i, j\leq N}
\\
\label{qdetgv37}
& = Z_q (N, L, {\CM}) = Z_q (L, {\CM}, N) .
\end{align}
\end{Theorem}

\begin{proof}
First, the Schur function \eqref{sch} is expressed as a determinant in terms of the complete homogeneous symmetric functions (the Jacobi--Trudi relation, \cite{macd, stan1,stan2}),
\begin{gather}\label{qdetgv31}
S_{\widehat{\BM}}\bigg(\frac{{\bf q}_{2N-k}}{q}\bigg)
= \det \bigg(h_{\CM-i+j} \bigg(\frac{{\bf q}_{2N-k}}{q} \bigg)\bigg)_{1\leq i, j\leq N},
\end{gather}
where the symmetric functions are expressed through the $q$-binomial coefficients (\ref{qanal1}):
\begin{gather}
\label{csf1}
h_{\CM-i+j} \bigg(\frac{{\bf q}_{2N-k}}{q} \bigg)
= \begin{bmatrix} N+L+\CM-i+j-1 \\ \CM-i+j \end{bmatrix}\! .
\end{gather}
Using (\ref{csf1}) and the Pascal formula for $q$-binomial coefficients \cite{stan1,stan2},
one transforms (\ref{qdetgv31}):
\begin{gather}
\label{qdetgv33}
S_{\widehat{\BM}}\bigg(\frac{{\bf q}_{2N-k}}{q}\bigg) =
\det \left(q^{(j-1)(\CM+j-i)}
\begin{bmatrix} \CM+N+L-i \\ N+L-j\end{bmatrix}
\right)_{1\leq i, j\leq N} \!.
\end{gather}
Eventually, equation~(\ref{qdetgv1}) is valid due to (\ref{wmdpf3}) and (\ref{qdetgv33}). The answer for the partition function $W_q (N, \CM)$ is obtained from (\ref{qdetgv1}) at $k=0$.

As a next step, we drop the powers of $q^{\CM}$, as well as the multiples $[{\CM}+L+N-i]$ and~$\frac{1}{[L+N-j]}$, $1\le i, j \le N$, out of rows and columns of the matrix in right-hand side of (\ref{qdetgv33}), and obtain:
\begin{gather}
\nonumber
S_{\widehat{{\BM}}}\bigg(\frac{{\bf q}_{2N-k}}{q}\bigg) = q^{\frac{{\CM}N}{2}
(N-1)} \prod_{i=1}^N \prod_{j=1}^{L} \frac{[{\CM}+i+j-1]}{[j+i-1]}
\label{qdetgv34}
\det \left(q^{(j-1)(j-i)}
\begin{bmatrix} {\CM}+N-i \\ N-j\end{bmatrix}
\right)_{1\leq i, j\leq N} \!.
\end{gather}
The determinant in (\ref{qdetgv34}) is calculated iteratively, and unity appears after $N-1$ steps as its value.
Therefore, equation~\eqref{qdetgv37} is valid due to \eqref{pf1} and \eqref{wmdpf3}.
\end{proof}

The determinantal representations (\ref{qanal6}) and (\ref{qdetgv1}) given by Theorems~\ref{theoremII}
and~\ref{TheoremIII}, res\-pectively, are equivalent and lead to the same double product expressions \eqref{rep31} and \eqref{qdetgv37}, which
are interpreted as the generating functions of watermelons.
The determinantal representation~(\ref{qdetgv1}) allows to express
in the limit $q\rightarrow 1$ the number $A(N, L, {\CM})$ of the watermelons with deviation
(i.e., the number of plane partitions in ${\cal B}(N, L, {\CM})$):
\begin{gather}
\label{wmdbd11}
A(N, L, {\CM})= \det \left(\begin{pmatrix} L+{\CM}+N-i \\ {\CM}+N-j \end{pmatrix} \right)_{1\leq i, j\leq N}=\prod_{i=1}^N \prod_{j=1}^{{\CM}} \frac{L+i+j-1}{j+i-1} .
\end{gather}
Equation (\ref{wmdbd11}) provides the statement of the Gessel--Viennot theorem \cite{ges} that connects the binomial determinant with the number of nests of self-avoiding lattice paths.

The matrix ${\overline T}^{\dl}$ \eqref{qanal55} at $\dl=k$ is simplified in the limit $M\to\infty$ so that the corresponding determinant is tractable. As a result, the partition function ${W}_q (N, L, {\CM})$ at $\dl=k$ acquires the form of the norm-trace generating function \cite{statm} of plane partitions with fixed traces of diagonal parts. The statement is given by the following

\begin{Theorem}\label{TheoremIV}
The partition function ${W}_q (N, L, {\CM})$ of the watermelon with deviation $N-L\allowbreak=k$ expressed by \eqref{volwaterk31} at $\dl=k$ takes the form at $M\to\infty$:
\begin{gather}
\label{qdetgv370}
{W}_q(N, L, {\CM})
\underset{M\to\infty}{= } q^{\frac{N-L}{2}(N+L-1)}
\prod_{i=1}^{L} \prod_{j=1}^{N}
\frac{1}{1-q^{k+i+j-1}} .
\end{gather}
\end{Theorem}

\begin{proof}
The entries \eqref{qanal55} at $\dl=k$ are simplified in the limit $M\to\infty$ provided that $q^{kM}\ll 1$:
\begin{gather*}
{\overline T}_{i j}^{\dl=k} = \displaystyle{\frac{1}{1- q^{k+j+i-1}}} ,\qquad
1\le i\le N-k ,\qquad 1\le j \le N .
\end{gather*}
In order to evaluate $\det {\overline T}^{(\dl=k)}$, one firstly combines neighboring rows $i^{\rm th}$ and $(i+1)^{\rm th}$ in~\eqref{qanal55}, $L+1 \le i < N$, as required to calculate the Vandermonde determinant. After this, the columns~$j^{\rm th}$ and $N^{\rm th}$, $1 \le j < N$,
are combined to obtain $N-1$ zeros in $N^{\rm th}$ row. After $N-L$ steps, one obtains:
\begin{align}
\nonumber
\det {\overline T}^{(\dl=k)} & = q^{\frac{N}{2} (N-1)+\frac{L}{2}(2N-L-1)(N-L)}
\prod_{i=1}^{L} \prod_{j=L+1}^{N}
\frac{1}{1-q^{k+i+j-1}}
\\
\label{qdetgv371}
& \times \frac{\CV({\bf q}_N/q)}{\CV({\bf q}_{L})}
\det \bigg(\frac{1}{1-q^{k+i+j-1}}\bigg)_{1\leq i, j\leq L}\! .
\end{align}
The Cauchy determinant in \eqref{qdetgv371} is evaluated \cite{statm}, and one obtains \eqref{qdetgv370} from
\eqref{volwaterk31} provided that \eqref{qdetgv371} is used in \eqref{werk12}.
\end{proof}

The representation \eqref{qdetgv371} enables the limiting form of the Schur function \eqref{werk12}.

\section[XX Heisenberg chain and dynamical correlation functions]
{$\boldsymbol{XX}$ Heisenberg chain and dynamical correlation functions} \label{randwalvic}

The Heisenberg $XX$ model on a chain of $M+1$ sites is defined by the Hamiltonian:
\begin{gather}
\label{xxham}
H_{XX} = \frac {\jum}2 + \frac 12 \sum_{k=0}^{M} (\BI - \si_k^z) ,
\end{gather}
where $\BI$ is identity operator, and the contribution
\begin{gather}
{\jum} = -\sum_{n, m =0}^M
\Dl_{n m} \si_n^{-} \si_m^{+}
\label{hamr}
\end{gather}
describes the neighbouring spins coupling through the {\it hopping matrix} ${\boldsymbol\Delta}$,
\begin{gather}
\label{qanal11}
{\bold\Delta}\equiv \bigl(\Dl_{n m} \bigr)_{0 \le n, m\le M} ,\qquad \Dl_{n m} = \dl_{|n-m|, 1} + \dl_{|n-m|, M} ,
\end{gather}
where $\dl_{k l}\equiv \dl_{k, l}$ is the Kronecker symbol. The entries of ${\bold\Delta}$ are subjected to the periodicity requirements: $\Dl_{n+M+1, m} = \Dl_{n m}$ and $\Dl_{n, m+M+1} = \Dl_{n m}$.
The local spin operators $\si^\pm_k = \frac12 (\si^x_k\pm{\rm i}\si^y_k)$ and $\si^z_k$ act nontrivially on $k^{\rm th}$ site and obey the comm\-ut\-at\-ion rules:
\begin{gather*}
[ \si^+_k, \si^-_l ] = \dl_{kl} \si^z_l ,\qquad
[ \si^z_k,\si^\pm_l ] = \pm 2 \dl_{kl} \si^{\pm}_l.
\end{gather*}
The spin operators
$\si^\pm_k$ and $\si^z_k$
act in the space ${\mathfrak H}_{M+1}$ spanned over the sta\-t\-es $\bigotimes_{l=0}^M
| s \rangle_{l}$, where $|
s\rangle_{l}$ implies either spin ``up'',
$\mid\uparrow\rangle$, or spin ``down'',
$\mid\downarrow\rangle$, state at $l^{\rm th}$ site. The states $\mid\uparrow\rangle\equiv
\left(\begin{smallmatrix}
1 \\0
\end{smallmatrix}\right)$ and $\mid\downarrow\rangle
\equiv
\left(\begin{smallmatrix}
0 \\1
\end{smallmatrix}\right)$ provide a natural basis of the linear space ${\BC}^2$. The Hamiltonian (\ref{xxham}) annihilates
the state $\mid\Uparrow\rangle
\equiv \bigotimes_{n=0}^M \mid \uparrow
\rangle_n$ with all spins ``up'', $H_{XX}\mid \Uparrow\rangle =0$, and commutes with the third component of the total spin ${S}^z$:
\begin{gather*}
\lbrack{H}_{XX}, {S}^z] = 0 ,\qquad
{S}^z\equiv \frac 12\sum_{k=0}^M\sigma_k^z .
\end{gather*}

Consider an arbitrary state on the chain characterized by $N$ spins ``down'' and ${\CK}\equiv M-N+1$ spins ``up''. The spin ``down'' sites are labelled by parts of a strict partition ${\bmu}=(\mu_1, \mu_2, \dots , \mu_N)$. We define the state $|\bmu \rangle$ corresponding to $N$ spins ``down''
and its conjugate $\langle \bnu |$,
\begin{gather}
\label{conwf01}
|\bmu \rangle \equiv
\Bigg(\prod\limits_{k=1}^N \si_{\mu_k}^{-}\Bigg) \mid\Uparrow \rangle ,\qquad
\langle \bnu | \equiv \langle \Uparrow \mid
\Bigg(\prod\limits_{k=1}^N \si_{\nu_k}^{+}\Bigg),
\end{gather}
which provide a complete orthogonal base:
\begin{gather}
\langle \bnu | \bmu \rangle =
\bdl_{\bnu \bmu} \equiv \prod\limits_{n=1}^N\dl_{\nu_n \mu_n} .
\label{tuc3}
\end{gather}

The $N$-particles state-vector $| \Psi({\bf u}_N)\rangle$ under \textit{off-shell} parametrization
${\bf u}_N$
is the combination of the states (\ref{conwf01}) with the Schur functions $S_\bla ({\textbf u}^{2}_N)$ (\ref{sch}) as the coefficients:
\begin{gather}
| \Psi({\textbf u}_N)\rangle = \sum\limits_{\bla \subseteq \{{\CK}^N\}}
S_\bla \big({\textbf u}^{2}_N\big) | \bmu \rangle. \label{conwf1}
\end{gather}
The conjugate off-shell state-vectors are given by
\begin{gather}\label{conj}
\langle \Psi({\bf v}_N) | = \sum\limits_{\bla \subseteq \{{\CK}^N\}}
\langle \bnu | S_\bla \big({\textbf v}^{-2}_N\big) .
\end{gather}
Summation in (\ref{conwf1}) and (\ref{conj}) goes over the partitions $\bla$ corresponding either to $\bmu$ or $\bnu$
(see Section~\ref{sectdet} and Figure~\ref{fig:f4}).

The periodic boundary conditions are imposed: $\si^{\#}_{k+(M+1)}=\si^{\#}_k$. If the
parameters $u^2_j\equiv
{\rm e}^{{\rm i}\ta_j}$ ($1\le j\le N$) satisfy
the Bethe equations \cite{col},
\begin{gather}
{\rm e}^{{\rm i} (M+1)\ta_j}=(-1)^{N-1} , \qquad 1\le j \le N ,
\label{betheexp}
\end{gather}
then the state-vectors (\ref{conwf1})
become the eigen-vectors of the Hamiltonian (\ref {xxham}):
\begin{gather}\label{egv}
{H}_{XX}| \Psi ({\bth }_N)\rangle = E_N({\bth }_N)| \Psi ({\bth }_N)\rangle ,
\end{gather}
where $\big| \Psi \big({\rm e}^{\frac {\rm i}2{\bth}_{N}}\big)\big\rangle$ is written
for simplicity as $| \Psi({\bth}_{N})\rangle$. The use of ${\bth}$ instead of ${\bth}_{N}$ is allowed if there is no confusion. The solutions $\theta_j$ to the Bethe equations (\ref{betheexp}) can be parametrized such that
\begin{gather}
\theta_j = \frac{2\pi}{M+1}\bigg(I_j - \frac{N-1}{2}\bigg),
\qquad 1\le j \le N , \label{besol}
\end{gather}
where $I_j$ are integers or half-integers depending on whether $N$ is odd or even.

The eigen-energy in (\ref{egv}) is given by
\begin{gather}
E_N({\bth}_N) = N - \sum_{j=1}^N\cos\ta_j = N -
\sum_{j=1}^N\cos \bigg(\frac{2\pi}{M+1}
\bigg(I_j-\frac{N-1}{2}\bigg)\bigg).
\label{egen}
\end{gather}
The \textit{ground state} of the model is the eigen-state determined by (\ref{besol}) at $I_j=N-j$:
\begin{gather*}
\theta^{ \rm g}_j \equiv \frac{2\pi}{M+1}\bigg(\frac{N+1}{2} -j\bigg), \qquad
1\le j \le N ,
\end{gather*}
and the corresponding eigen-energy
\[
E_N(\bth^{ \rm g}_N) = N-\csc \bigg(\frac{\pi}{M+1}\bigg) \sin \bigg(\frac{\pi N}{M+1}\bigg).
\]
The squared norm of the state vectors (\ref{conwf1}) on the solutions (\ref{besol}) is
\begin{gather}\label{norma}
{\CN}^2({\bth}_N) \equiv \langle
\Psi ({\bth}_N)
\mid\!\Psi ({\bth}_N)\rangle =
\frac{(M+1)^N}{|\CV ({\rm e}^{{\rm i}{\bth}_N})|^2} .
\end{gather}

Let us introduce two operators,
\begin{gather}
\label{fp}
\bar\varPi_{n} \equiv \prod\limits_{j=0}^{n-1} \frac{\BI + \si^z_j}{2},\qquad
\bar{\sf F}_{n} \equiv \prod\limits_{j=0}^{n-1} \si^-_j,
\end{gather}
where $\bar\varPi_{n}$ is the {\it projection operator} forbidding the spin ``down'' states on $n$ successive sites,
and~$\bar{\sf F}_{n}$ is the \textit{domain wall operator} that creates the spin ``down'' states on $n$ successive sites. It~is assumed that $\bar\varPi_{0}$ and $\bar{\sf F}_{0}$ are identity operators.

Let us define the ``dynamical'' \textit{auto-correlation} function, which generalizes the correlation functions considered in \cite{bmnph}:
\begin{align}\label{auto}
\Gamma\big({\bth}_{[N\backslash n]}^{ \rm g}, N, n, m, t_1, t_2\big)
\equiv \frac{\big\langle
\Psi \big({\bth}^{ \rm g}_{[N\backslash n]}\big)\big| \bar{\sf F}_{n}^{+} {\rm e}^{-t_1 H_{XX}} {\bar\varPi}_m {\rm e}^{-t_2 H_{XX}} \bar{\sf F}_{n} \big| \Psi \big({\bth}^{ \rm g}_{[N\backslash n]}\big) \big\rangle}
{\big\langle \Psi \big({\bth}^{ \rm g}_{[N\backslash n]}\big)\big| \Psi \big({\bth}^{ \rm
g}_{[N\backslash n]}\big)\big\rangle} ,
\end{align}
where the operators are given by (\ref{fp}), $\bar{\sf
F}^{+}_{n}$ is the Hermitian conjugated operator acting on (\ref{conj}), and
$t_1$, $t_2$ are the \textit{evolution} parameters.
The states \eqref{conwf1} and \eqref{conj} in \eqref{auto}
are parameterized by elements of
$(N-n)$-tuple ${\bth}_{[N\backslash n]}^{ \rm g} \equiv \big(\theta^{ \rm g}_{n+1}, \theta^{ \rm g}_{n+2}, \dots, \theta^{ \rm g}_{N}\big)$ consisting of the ground state solutions to equation~(\ref{betheexp}):
\begin{gather}\label{grstxxn}
\theta^{ \rm g}_j \equiv \frac{2\pi}{M+1}\bigg(\frac{N-n+1}{2} -j\bigg), \qquad
1\le j \le N-n .
\end{gather}
The \textit{persistence of domain wall} correlation function arises
from (\ref{auto}) at $m=0$ and $t=t_1+t_2$:
\begin{gather}
{\cal F} \big({\bth}^{ \rm g}_{[N\backslash n]}, n, t\big) \equiv \frac{\big\langle
\Psi \big({\bth}^{ \rm g}_{[N\backslash n]}\big) \big| \bar{\sf
F}^{+}_{n} {\rm e}^{-t H_{XX}} \bar{\sf F}_{n}
\big| \Psi \big({\bth}^{ \rm g}_{[N\backslash n]}\big)\big\rangle
}{\big\langle \Psi \big({\bth}^{ \rm g}_{[N\backslash n]}\big) \big|
\Psi \big({\bth}^{ \rm g}_{[N\backslash n]}\big)\big\rangle} ,
\label{field0}
\end{gather}
and ${\cal F} \big({\bth}^{ \rm g}_{[N\backslash n]}, 0, t\big)=1$ since $\bar{\sf F}_{0}$ is identity operator.

\section[Random turns vicious walkers and correlations over ferromagnetic state]
{Random turns vicious walkers and correlations\\ over ferromagnetic state}
\label{randfer}

In order to interpret the auto-correlation function (\ref{auto}) as a sum over nests of self-avoiding lattice paths, we shall firstly pay attention to the correlations over the ferromagnetic state.

Let us turn to description of $N$ \textit{vicious walkers} \cite{fish, forr1, 2}, which are located initially on a chain at the positions
$l_1 > l_2> \cdots > l_N$. Two essentially different
types of walks of vicious walkers may be distinguished: \textit{random turns} and \textit{lock step} walks. Occupation of a site by two walkers is forbidden. In the lock step model at each tick of the
clock each walker jumps, with equal probability, to the left or to the right site (see, for instance, Figure~\ref{fig:f7}, where paths connect~$C_i$ to~$B_i$, $1\le i\le 4$). In the random turns model at each tick only a single randomly chosen walker jumps to
right or to left site, whereas the rest are staying (Figure~\ref{fig:f18}). After $K$ steps, the walkers arrive at the positions $j_1>j_2>\cdots>j_N$.
Trajectories
of the random walkers can be viewed as \textit{directed lattice paths} (i.e., the paths
that can not turn back) on the square lattice.
\begin{figure}[h]
\center
\includegraphics[scale=.9]{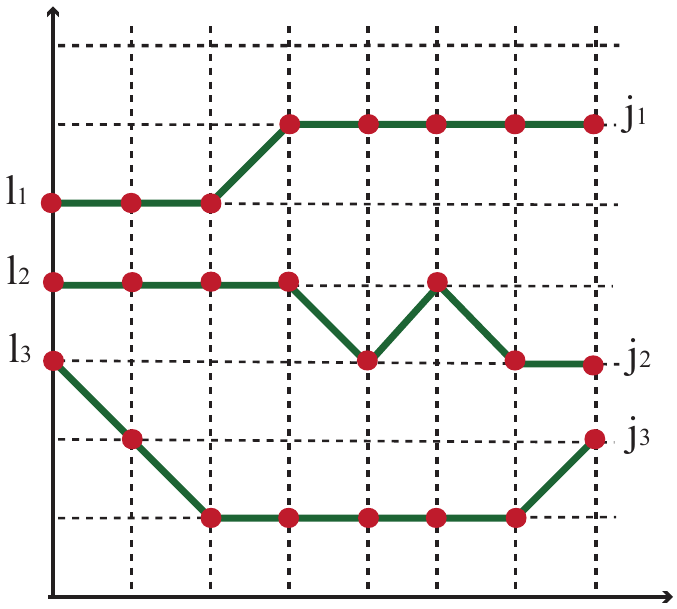}
\put(1,4){\makebox(0,0)[lb]{$t$}}
\caption{Trajectories of $K=7$ steps for $N=3$ random turns vicious walkers.}
\label{fig:f18}
\end{figure}

We define $N$-particles \textit{transition amplitude}:
\begin{gather}\label{mpcf}
{\cal G}({\bf j}; {\bf l} | t) \equiv \big\langle {\bf j} \big| {\rm e}^{- t H_{XX}} \big| {\bf l} \big\rangle ,
\end{gather}
where $H_{XX}$ is the Hamiltonian (\ref{xxham}), and the states $\langle {\bf j} |$ and $| {\bf l} \rangle$ given by (\ref{conwf01}) are respectively labelled by the strict partitions ${\bf j} = (j_1, j_2, \dots, j_N)$ and ${\bf l} = (l_1, l_2, \dots, l_N)$. The generating function of self-avoiding lattice paths of the vicious walkers is enabled by the transition amplitude~(\ref{mpcf}):
\begin{gather}\label{mpgf}
G({\bf j}; {\bf l} | t) \equiv
\big\langle {\bf j} \big| {\rm e}^{- \frac t2 {\jum}} \big| {\bf l}\big \rangle = {\rm e}^{t N} {\cal G}({\bf j}; {\bf l} | t) ,
\end{gather}
where ${\jum}$ is given by (\ref{hamr}). The condition that vicious walkers do not touch each other up to $N$ steps is guaranteed by the property of the Pauli matrices $(\sigma_k^{\pm})^2=0$.

Differentiating $G ({\bf j}; {\bf l} | t)$ (\ref{mpgf}) with respect of $t$ and applying the commutation relation\vspace{-1ex}
\begin{gather}
\big[ {\jum},\si_{l_1}^{-} \si_{l_2}^{-} \cdots \si_{l_N}^{-}\big] = \sum_{k=1}^N \si_{l_1}^{-} \cdots \si_{l_{k-1}}^{-} [{\jum},\si_{l_k}^{-}] \si_{l_{k+1}}^{-} \cdots\si_{l_N}^{-} , \label{comcom}
\end{gather}
we obtain the equation\vspace{-1ex}
\begin{gather}
\frac{\rm d}{{\rm d}t} G ({\bf j}; {\bf l} | t) = \frac 12 \sum_{k=1}^N\bigl(G({\bf j}; l_1, l_2, \dots , l_k+1, \dots, l_N | t)
+ G({\bf j}; l_1, l_2, \dots, l_k-1, \dots, l_N | t)\bigr)
\label{phem}
\end{gather}
{\samepage
(the ``final'' position ${\bf j}$ is fixed), and a similar equation can be found for fixed
${\bf l}$. The commutation relation\vspace{-1ex}
\begin{gather}
\lbrack {\jum},\si_m^{-}]= - \sum_{n=0}^M \Dl_{n m}\si_n^{-} \si_m^z , \label{cr}
\end{gather}
is also used in derivation of (\ref{phem}).

}

The non-intersection condition means that $G ({\bf j}; {\bf l} | t)=0$ if $l_k=l_p$ (or $j_k=j_p$) for any $1\leq
k, p\leq N$. Equation (\ref{phem}) is supplied with the periodicity requirement in each $j_k$, $m_k$ with all other $j_s$, $m_s$ ($k\neq s$) fixed:
\begin{gather*}
 G(j_1, j_2, \dots, j_k+M+1, \dots, j_N; l_1, l_2, \dots, l_N|t)
 \\ \qquad
{}= G(j_1, j_2, \dots, j_N; l_1, l_2, \dots, l_k+M+1, \dots, l_N|t) = G(j_1, j_2, \dots, j_N; l_1, l_2, \dots, l_N|t) .
\end{gather*}
The periodic solution to (\ref{phem}) subjected to the ``initial condition'' $G({\bf j}; {\bf l} | 0)= \bdl_{{\bf j} {\bf l}}$ (see (\ref{tuc3})) is given by

\begin{Proposition}\label{Proposition5}
The
determinantal representation for solution to equation~\eqref{phem} is valid:
\begin{align}
G(j_1, j_2, \dots,j_N; l_1, l_2, \dots, l_N|t) &= \det \bigl( G(j_r, l_s|t)\bigr)_{1\le r,s \le N}
\nonumber
\\
&= \det \bigl(\big({\rm e}^{\frac{t}2 \bold \Delta}\big)_{j_r, l_s} \bigr)_{1\le r,s \le N} , \label{det}
\end{align}
where $G(j, l|t)$ is given by
\begin{gather}
{G} (j, l | t )\equiv \big\langle \Uparrow \big| \si_j^{+} {\rm e}^{- \frac t2 {\jum}} \si_l^{-} \big| \Uparrow\big \rangle
= \bigr({\rm e}^{\frac t2{\bold\Delta}}\bigl)_{j l} .
\label{cfff}
\end{gather}
\end{Proposition}

\begin{proof}
The representation equation~(\ref{det}) is derived in \cite{circ}.
\end{proof}

The determinantal formulae for transition probabilities of $N$ particles are well-known in probability and combinatorics for a long time \cite{ges,kar1,kar, lind} and still attract attention, e.g., \cite{kat, kat1}.

The unitary matrix $\widehat{\bf u}\equiv(M+1)^{-1/2}\bigl({\rm e}^{{\rm i} l\phi_s}\bigr)_{0\le s, l \le M}$, where $\phi_s=\frac{2\pi}{M+1}\bigl(s-\frac{M}{2} \bigr)$,
diagonali\-zes~${\bold\Delta}$~(\ref{qanal11}). Therefore the entries (\ref{cfff}) are:
\begin{gather}
G(j, l | t) \equiv \frac{1}{M+1} \sum\limits_{s=0}^M
{\rm e}^{{\rm i} \phi_s(j-l)+ t \cos\phi_s}.
\label{ratbe61}
\end{gather}
Using (\ref{cfff}), (\ref{ratbe61}) in (\ref{det}), one obtains another equivalent determinantal solution to (\ref{phem}) given~by

\begin{Proposition}\label{Proposition6}
Solution to \eqref{phem} is of the form:
\begin{gather}
G({\bf j}; {\bf l} | t) =
\frac{{\rm e}^{t N}}{(M+1)^N}\sum\limits_{\{{\bphi}_N\}}
{\rm e}^{- t E_N({\bphi}_N)}\big|{\CV} ({\rm e}^{{\rm i} {\bphi}_N}\big)\big|^2
S_{{\bla^L}}\big({\rm e}^{{\rm i} {\bphi}_N}\big)
S_{{\bla^R}}\big({\rm e}^{-{\rm i} {\bphi}_N}\big),
\label{ratbe7}
\end{gather}
where the sum is over $N$-tuples ${\bphi}_N = (\phi_{k_1}, \phi_{k_2}, \dots, \phi_{k_N})$, $\phi_n = \frac{2\pi}{M+1} \bigl(n-\frac{M}{2}\bigr)$ and $M\ge k_1>k_2\allowbreak > \cdots > k_N\ge 0$. The partitions ${\bla^L}$, ${\bla^R}$ and ${\bf j}$, ${\bf l}$ are related:
${\bla^L}={\bf j}- {\bdl}_N$, ${\bla^R}={\bf l}- {\bdl}_N$, and ${\bdl}_N$ is~\eqref{stair}. Besides, the Vandermonde deteminant ${\CV} ({\rm e}^{{\rm i} {\bphi}_N})$, where ${\rm e}^{\pm{\rm i} {\bphi}_N}\equiv ({\rm e}^{\pm{\rm i}\phi_1}, {\rm e}^{\pm{\rm i}\phi_2}, \dots, {\rm e}^{\pm{\rm i}\phi_N})$, is defined by \eqref{spxx1}, and $E_N({\bphi}_N)$ is defined by \eqref{egen}.
\end{Proposition}

\begin{proof}
The representation equation~(\ref{ratbe7}) is derived in \cite{bmnph}.
\end{proof}

Expanding the representation
(\ref{mpgf}) in powers of $t$, we obtain the series
\begin{gather}\label{avv}
G({\bf j}; {\bf l} | t) = \sum_{K=0}^{\infty}
\frac{(t/2)^K}{K!} \mathfrak{G} ({\bf j}; {\bf l} | K) ,
\end{gather}
where the coefficients are the averages:
\begin{gather}
\label{qanal13}
\mathfrak{G}({\bf j}; {\bf l} | K) \equiv
\big\langle {\bf j} \big| (- {\jum})^K \big| {\bf l} \big\rangle .
\end{gather}
Let us define $N$-tuples ${\bf e}_{k}$, $1\le k\le N$, consisting of zeros except a unity at $k^{\rm th}$ place (say, from left). The commutation relation (\ref{cr}) allows to establish that $\mathfrak{G} ({\bf j}; {\bf l} | K)$ (\ref{qanal13}) satisfies the identity:
\begin{gather}
\mathfrak{G} ({\bf j}; {\bf l} | K+1) = \sum_{k=1}^{N} \bigl(\mathfrak{G} ({\bf j}; {\bf l}+{\bf e}_{k} | K) +
\mathfrak{G} ({\bf j}; {\bf l}-{\bf e}_{k} | K)\bigr) .
\label{mpcf6}
\end{gather}
Solution to (\ref{mpcf6}) takes the form \cite{circ}:
\begin{gather}
\label{kst}
\mathfrak{G} (j_1, j_2, \dots, j_N; l_1, l_2, \dots, l_N|K) = \sum_{|{\bf k}|=K} P({\bf k}) \det \bigl(\big( \bold\Delta^{k_r}\big)_{j_r, l_s}\bigr)_{1 \le r,s \le N} ,
\end{gather}
where $|{\bf k}| \equiv k_1+k_2+\cdots+k_N$, and $P({\bf k})$ is the multinomial coefficient:
\[
P({\bf k}) \equiv \frac{(k_1+k_2+ \cdots + k_N)!}{k_1!k_2!\cdots k_N!} .
\]

Let $|P_K ({\bf l} \rightarrow {\bf j})|$ be the number of $K$-step paths traced by $N$ vicious walkers in the random turns model \cite{fish}. From (\ref{kst}) it is clear, \cite{circ}, that equation~(\ref{qanal13}) provides us the number of configurations of $N$ random turns walkers being initially located on the sites $l_1>l_2> \dots > l_N$ and arrived after $K$ steps at the positions $j_1>j_2> \cdots > j_N$:
\begin{gather}
\label{qanal27}
|P_K ({\bf l} \rightarrow {\bf j})| = \mathfrak{G}({\bf j}; {\bf l} | K) .
\end{gather}

The $N$-particles transition amplitude (\ref{mpcf}) arises from (\ref{mpgf}), where (\ref{det}) or (\ref{ratbe7}) are used.

\section{Two-time correlations over ferromagnetic state}\label{twopoint}

Before studying (\ref{auto}), let us consider, by analogy with ${\cal G} ({\bf j}; {\bf l} | t)$ (\ref{mpcf}) and $G({\bf j}; {\bf l} | t)$ (\ref{mpgf}), ``two-time'' transition amplitude,
\begin{gather}
{\cal G}({\bf j}; {\bf l} | t_1, t_2, m) \equiv \big\langle {\bf j} \big|
{\rm e}^{-{t_1} H_{XX}} {\bar\varPi}_m {\rm e}^{- {t_2} H_{XX}} \big| {\bf l} \big\rangle ,
\label{twop777}
\end{gather}
and ``two-time'' generating function of self-avoiding lattice paths:
\begin{gather}
{G} ({\bf j}; {\bf l} | t_1, t_2, m )\equiv \big\langle {\bf j} \big|
{\rm e}^{-\frac{t_1}2 {\jum}} {\bar\varPi}_m {\rm e}^{- \frac{t_2}2 {\jum}} \big| {\bf l} \big\rangle = {\rm e}^{(t_1+t_2)N} {\cal G} ({\bf j}; {\bf l} | t_1, t_2, m) .
\label{twop7}
\end{gather}
As in Section~\ref{randfer},
differentiating (\ref{twop7}) with respect to $t_1$ and $t_2$ and using the commutation relations (\ref{comcom}), (\ref{cr}) we obtain the difference-differential equation:
{\samepage\begin{gather}
 \frac{\partial^2}{\partial t_1 \partial t_2} {G} ({\bf j}; {\bf l} | t_1, t_2 ,m )
\nonumber
\\ \qquad
{} = \frac{1}{4} \sum_{f, p =1}^{N}
\sum_{s, s^{\prime} = \pm 1} {G} (j_1, \dots, j_f + s, \dots , j_N; l_1, \dots, l_p + s^{\prime}, \dots, l_N | t_1, t_2 ,m ).
\label{twop9}
\end{gather}
Equation (\ref{twop9}) is supplied with the ``initial condition'': $G ({\bf j}; {\bf l} | 0, 0, m) =
\prod_{p=1}^N \dl_{j_p, l_p}$ provided that $j_p, l_p \ge m$, $\forall p$.

}

Let us first consider
the simplest ``two-time''
generating function:
\begin{gather}
{G} (j, l | t_1, t_2, m)\equiv \big\langle \Uparrow \big| \si_j^{+}
{\rm e}^{-\frac{t_1}2 {\jum}} {\bar\varPi}_m {\rm e}^{- \frac{t_2}2 {\jum}} \si_l^{-} \big| \Uparrow \big\rangle .
\label{twop1}
\end{gather}
We derive the
difference-differential equation differentiating (\ref{twop1}) and applying the commutation relation (\ref{cr}):
\begin{align}
\frac{\partial^2}{\partial t_1 \partial t_2} {G} (j, l | t_1, t_2, m ) = {}&\frac{1}{4} \bigl({G} (j+1, l+1 | t_1, t_2, m) + {G} (j+1, l-1 | t_1, t_2, m)
\nonumber
\\
& + {G} (j-1, l+1 | t_1, t_2, m ) + {G} (j-1, l-1 | t_1, t_2, m )\bigr).
\label{twop3}
\end{align}
As far as the ``initial'' conditions are concerned, we obtain from (\ref{twop1}) that ${G} \left(j, l | t_1, t_2, m \right)$ vani\-shes
at $t_1=t_2 =0$ for $j$, $l$ respecting
$0\le j, l < m$. Otherwise, ${G} (j, l | 0, 0, m)=\dl_{j l}$.
The solution to (\ref{twop3}) is of the form
\begin{gather}
\label{twop5}
{G} (j, l | t_1, t_2, m) = \sum_{k=m}^{M}{G} (j, k | t_1) {G} (k, l | t_2),
\end{gather}
where ${G} (j, l | t)$ is the solution (\ref{ratbe61}). We obtain from
(\ref{twop5}) at $m=0$:
\begin{gather*}
{G} (j, l | t_1, t_2, 0 ) =
{G} (j, l | t_1 + t_2 ) ,
\end{gather*}
in agreement with
(\ref{cfff}) and (\ref{twop1}).

Regarding (\ref{twop5}), we arrive at

\begin{Proposition}\label{Proposition7}
Solution to equation~$(\ref{twop9})$ takes the form:
\begin{gather}
{G} ({\bf j}; {\bf l} | t_1, t_2, m ) = \sum_{\bro \subseteq\{({\cal M}\slash m)^N\}}
{G} ({\bf j}; {\bro}+{\bdl}_N | t_1 ){G} ({\bro}+{\bdl}_N; {\bf l} | t_2 ) ,
\label{twop11}
\end{gather}
where ${G} ({\bf j}; {\bro}+{\bdl}_N | t_{1, 2})$ are given by $(\ref{ratbe7})$.
\end{Proposition}

\begin{proof} It is straightforward to verify that solution to (\ref{twop9}) can be written as
\begin{gather}
{G} ({\bf j}; {\bf l} | t_1, t_2, m) = \det ({G} (j_k, l_n | t_1, t_2, m))_{1\le k, n\le N} ,
\label{twop12}
\end{gather}
where the entries ${G} (j_k, l_n | t_1, t_2, m)$ given by (\ref{twop5}) correspond to the product of two rectangular matrices. Application of the Cauchy--Binet formula for matrices enables to arrive from (\ref{twop12}) to~(\ref{twop11}).
Taking into account that the entries ${G} (j_k, l_n \mid t)$ are of the form (\ref{ratbe61}),
one arrives at~(\ref{twop11}), where ${G} ({\bf j}; {\bf l} \mid t)$ is given by (\ref{ratbe7}).
\end{proof}

Expanding (\ref{twop7}) in double series with respect of $\bigl(\frac{t_1}2 \bigr)^{K_1}$ and $\bigl(\frac{t_2}2 \bigr)^{K_2}$, one obtains the coefficients expressed by the average:
\begin{gather}
\label{twop19}
{\mathfrak G} ({\bf j}; {\bf l} | K_1, K_2, m) \equiv
\big\langle {\bf j} \big|(- {\jum})^{K_1} {\bar\varPi}_m (- {\jum})^{K_2} \big| {\bf l}\big\rangle .
\end{gather}
Applying (\ref{cr}), we obtain that ${\mathfrak G} ({\bf j}; {\bf l} | K_1, K_2, m)$ (\ref{twop19}) satisfies the identity:
\begin{gather}
{\mathfrak G} ({\bf j}; {\bf l} | K_1+1, K_2+1, m) =
\sum_{f, p =1}^{N}
\sum_{s, s^{\prime} = \pm 1} {\mathfrak G} ({\bf j}+ s {\bf e}_f; {\bf l}+ s^{\prime} {\bf e}_p | K_1, K_2, m ) ,
\label{twop21}
\end{gather}
where ${\bf e}_{p, f}$ are the same as in (\ref{mpcf6}).
Using (\ref{avv}) in (\ref{twop11}), we obtain the representation:
{\samepage\begin{gather}
{\mathfrak G} ({\bf j}; {\bf l} | K_1, K_2, m) = \sum_{\bro \subseteq
\{({\cal M}\slash m)^N\}}
{\mathfrak G} ({\bf j}; {\bro}+{\bdl}_N | K_1)
{\mathfrak G} ({\bro}+{\bdl}_N;
{\bf l} | K_2) .
\label{twop111}
\end{gather}
Since ${\mathfrak G} ({\bf j};
{\bf l} | K)$ (\ref{qanal13})
fulfils (\ref{mpcf6}), it is straightforward to check that (\ref{twop111}) fulfils (\ref{twop21}).

}

As in the case of equation~(\ref{qanal13}), the average ${\mathfrak G} ({\bf j}; {\bf l} | K_1, K_2, m)$ (\ref{twop19}) enumerates
the nests of $(K_1 + K_2)$-step paths
connecting the sites constituting ${\bf j}$ and ${\bf l}$:
\begin{gather}
\label{twop20}
\big|P_{K_1 + K_2} ({\bf l}\underset{m} \rightarrow {\bf j} )\big| \equiv {\mathfrak G} ({\bf j} ; {\bf l} | K_1, K_2, m) ,
\end{gather}
where the notation in left-hand side is a modification of
the notation (\ref{qanal27}).
Indeed, equation~(\ref{twop111}) is re-expressed regarding (\ref{qanal27}) and (\ref{twop20}) as follows:
\begin{gather}
\big|P_{K_1 + K_2} ({\bf l} \underset{m} \rightarrow {\bf j} )\big| = \sum_{\bro \subseteq
\{({\cal M}\slash m)^N\}}
|P_{K_1} ({\bf l} \rightarrow {\bro}+{\bdl}_N ) |
|P_{K_2} ({\bro}+{\bdl}_N \rightarrow {\bf j} )| .
\label{twop112}
\end{gather}
The representation (\ref{twop112}) demonstrates that (\ref{twop20}) enumerates the nests of such paths, which are due to ``gluing'' of the directed paths ${\bf l} \rightarrow {\bro}+{\bdl}_N$ to the directed paths ${\bro}+{\bdl}_N \rightarrow {\bf j}$, as
depicted in Figure \ref{fig:f18}. The ``gluing'' is along a ``dissection'' line characterized by the strict partition ${\bro}+{\bdl}_N$ provided that visiting the sites from $0^{\rm th}$ to $(m-1)^{\rm th}$ is forbidden. Therefore, $\big|P_{K_1 + K_2} ({\bf l} \underset{m} \rightarrow {\bf j} )\big|$~(\ref{twop20})
enumerates the random turns walks constrained by a \textit{bottleneck}. The bottleneck is absent if $m=0$, and (\ref{twop111}) is reduced to the coefficient corresponding to $K=K_1+K_2$ of the correlation function (\ref{det}) with $t=t_1+t_2$.

\section[Correlations over N-particles ground state]
{Correlations over $\boldsymbol{N}$-particles ground state}\label{npart}

This section is devoted to $N$-particles correlation functions (\ref{auto}) and (\ref{field0}).

\subsection[Form-factors of the operators bar Pi m and bar F n]
{Form-factors of the operators $\boldsymbol{{\bar\varPi}_{m}}$ and $\boldsymbol{\bar{\sf F}_{n}}$}

Form-factor of the projector $\bar\varPi_{m}$ (\ref{fp}) is defined as the average
\begin{gather}
\CT ({\textbf v_N}, {\textbf u_N}, m) \equiv \big\langle \Psi ({\textbf
v_N}) \big| \bar\varPi_{m} \big| \Psi ({\textbf u_N})\big\rangle , \label{ratio}
\end{gather}
where ${\textbf u_N}$ and ${\textbf v_N}$ are $N$-tuples of arbitrary complex numbers. Equations (\ref{conwf1}) and (\ref{conj}) allow us to calculate (\ref{ratio}),
and the answer is obtained in terms of the Cauchy--Binet identity (\ref{scschr}):
\begin{gather}
\CT ({\textbf v_N}, {\textbf u_N}, m) = \CP_{{\CK}/m}\big({\textbf v}^{-2}_N, {\textbf u}^2_N\big) =
\Bigg(\prod\limits_{l=1}^{N} \frac{u_l^{2m}}{v_l^{2m}}\Bigg)
\frac{\det(T_{k j})_{1\leq k, j\leq N}}{{\CV} ({\textbf u}_N^2){\CV} ({\textbf v}^{-2}_N)} ,
\label{spdfpxx0}
\end{gather}
where $T_{k j}$ are given by (\ref{tt}) with ${\cal L}$ replaced by ${\CK}$.
Equation~(\ref{spdfpxx0}) (together with (\ref{gener1})) demonstrates that (\ref{ratio}) is interpreted as a sum of the watermelons defined by~Construction~\ref{Construction1} (Section~\ref{waterm}). However, the watermelons are \textit{constrained} since summation in $\CP_{{\CK}/m}$ \eqref{spdfpxx0}
is restricted from below,
$\mu_N\ge m$, and the lattice paths are infiltrating above
the column of height~$m$ in the middle of Figure~\ref{fig:f7}($a$). For $m=0$, equation~(\ref{spdfpxx0}) expresses the scalar product.

Under the $q$-parameterization
\begin{gather}
{\bf v}_N^{-2}={\bf q}_N\equiv \big(q, q^2, \dots,q^N\big) ,\qquad
{\bf u}_N^{2}={\bf q}_N/q = \big(1, q, \dots, q^{N-1}\big) , \label{rep2111}
\end{gather}
the form-factor (\ref{ratio}) takes the form:
\begin{align}
\CT \big({\textbf{q}_N^{-\frac{1}{2}}}, (\textbf{q}_N/q)^{\frac{1}{2}}, m\big)
& = \CP_{\CK / m} (\textbf{q}_N, \textbf{q}_N/q)= q^{mN^2}W_q (N, {\cal M}-m)
\nonumber
\\
&= q^{mN^2} Z_q (N, N, {\CK}-m) , \label{teproj}
\end{align}
where the definition (\ref{gfw})
and Proposition~\ref{Proposition3} (see (\ref{volwater1}), (\ref{volwater2})) are used together with (\ref{qanal7}) and (\ref{scqschr12}). It is seen from (\ref{teproj}) that $q$-parametrized form-factor of
$\bar\varPi_{m}$ is the generating function $W_q (N, {\cal M}-m)$ of constrained watermelons or, equivalently, is the generating function $Z_q (N, N, {\CK}-m)$ of boxed plane partitions in ${\cal B}(N, N, {\CK}-m)$ (see \eqref{npp} and Figure \ref{fig:f7}):
\begin{gather}
\lim\limits_{q\to 1} {\CT} \big({\textbf{q}}_N^{-\frac{1}{2}}, (\textbf{q}_N/q)^{\frac{1}{2}}, m\big) = A (N, N, {\CK}-m) .
\label{rep55}
\end{gather}

Let us turn to the form-factor of the domain wall annihilation operator $\bar{\sf F}_{n}^+$ (\ref{fp}):
\begin{gather}
{\cal F} \big(\textbf v_{[N\backslash n]}, \textbf u_{N}, n\big) \equiv \big\langle \Psi \big({\bf v}_{[N\backslash n]}\big) \big| \bar{\sf
F}_{n}^+ \big| \Psi ({\bf u}_{N})\big\rangle , \label{ratiof}
\end{gather}
where (\ref{xindi1}) is taken into account. First, let us act by $\bar{\sf F}_{n}^+$ (\ref{fp}) on the state (\ref{conj}):
\begin{align}
\label{fstv2}
\big\langle \Psi\big({\bf v}_{[N\backslash n]}\big) \big| \bar{\sf F}_{n}^+
&= \sum_{{\bet} \subseteq \{({\CK}+n)^{N-n}\}} \big\langle {\bnu}_{N-n} \big| S_{{\bet}} \big({\textbf v}^{-2}_{[N\backslash n]}\big) \bar{\sf F}_{n}^+
\\
&= \sum_{{\bet} \subseteq \{{(({\CK}+n)/n)}^{N-n}\}} \big\langle {{\bnu}_{N-n}, \bdl_n }\big| S_{{\bet} } \big({\textbf v}^{-2}_{[N\backslash n]}\big)
\label{fstv1} \\
&= \Bigg(\prod\limits_{l=n+1}^{N} v^{-2n}_l \Bigg) \sum_{{\bet} \subseteq \{{\CK}^{N-n}\}}\big\langle {\widehat\bnu} \big| S_{{{\bet}}} \big({\textbf v}^{{-2}}_{[N\backslash n]}\big) ,
\label{fstv}
\end{align}
where ${\widehat\bnu}=
({\bnu}_{N-n}+{\bf n}_{N-n}, \bdl_n)$ (see (\ref{strhat})). The parts of $\bnu_{N-n}$ respect $M\ge\nu_1>\nu_2>\cdots > \nu_{N-n}\ge 0$ in (\ref{fstv2}),
or $M\ge\nu_1>\nu_2>\cdots > \nu_{N-n}\ge n$ in (\ref{fstv1}),
or $M-n\ge\nu_1>\nu_2>\cdots > \nu_{N-n}\ge 0$ in (\ref{fstv});
in all three cases summation is over non-strict partitions $\bet_{N-n}=\bnu_{N-n} -\bdl_{N-n}$. The state in (\ref{fstv}) and its
conjugate are defined:
\begin{gather}
\big\langle {\widehat{\bnu}} \big| \equiv \langle \Uparrow \mid
\Bigg(\prod\limits_{l=1}^{N-n}\si_{n+ \nu_l}^{+}\Bigg)
\prod\limits_{l=0}^{n-1} \si_{l}^{+} ,\qquad
\big| {\widehat{\bmu}} \big\rangle\equiv
\prod\limits_{p=1}^{N-n} \si_{n + \mu_p}^{-}
\Bigg(\prod\limits_{l=0}^{n-1} \si_{l}^{-}\Bigg)\mid \Uparrow \rangle . \label{ratbe2}
\end{gather}

One obtains from (\ref{fstv}) the transition amplitude interpreted as a sum of the stars with deviation $n$ (Figure \ref{fig:f555}):
\begin{gather}
\label{mel}
\big\langle \Psi\big({\bf v}_{[N\backslash n]}\big) \big| \bar{\sf F}_{n}^+ | \widehat{\bmu} \rangle =
\Bigg(\prod \limits_{l=n+1}^{N} v^{-2n}_l \Bigg) S_{{\bla}}\big({\bf v}^{-2}_{[N\backslash n]}\big) =
\sum_{\{{\cal C}_n \}} \prod_{j=n+1}^{N} v^{-2 ( c_j +n)}_{j} ,
\end{gather}
where (\ref{schdevi}) is taken into account, and ${\bla}\equiv {\bla}_{N-n}$ is given by $\widehat\bla$ in accord with (\ref{qanal14}), (\ref{strhat}).
The transition amplitude of the domain wall creation operator
$\bar{\sf F}_{n}$ calculated with the help of (\ref{limsch}) and (\ref{lim}) is a sum of the stars in Figure \ref{fig:pathconj}:
\begin{gather}
\label{mel2}
\big\langle \widehat{\bmu} \big| \bar{\sf F}_{n} \big| \Psi({\bf u}_{N-n}) \big\rangle =
\Bigg( \prod\limits_{l=1}^{N-n} u^{2n}_l \Bigg) S_{{\bla}} \big(\mathbf{u}_{N-n}^2\big) =
\sum_{\{{\cal B}_n \}}\prod_{l=1}^{N-n} u_l^{2(b_l+n)} ,
\end{gather}
where ${\bla}$ and $\widehat{\bmu}$ are the same as in (\ref{mel}).

 The averages (\ref{mel}) and (\ref{mel2}) are specified
for $|\widehat{\bmu} \rangle =
\big|\widehat{\bdl}_N \big\rangle$, $\langle \widehat{\bmu} | = \big\langle \widehat{\bdl}_N \big|$:
\begin{gather}
\big\langle \Psi\big({\bf v}_{[N\backslash n]}\big) \big| \bar{\sf F}_{n}^+ \big| \widehat{\bdl}_N \big\rangle = \prod_{j=n+1}^{N} v^{-2 n}_{j} , \qquad
\big\langle \widehat{\bdl}_N \big| \bar{\sf F}_{n} \big| \Psi({\bf u}_{N-n}) \big\rangle =
\prod_{l=1}^{N-n} u_l^{2n} .
\label{mel31}
\end{gather}
The stars, Figure \ref{fig:f555} or Figure \ref{fig:pathconj}, correspond to ratios of (\ref{mel})
or (\ref{mel2}) to appropriately chosen average (\ref{mel31}).

Using the definitions of the state-vectors
(\ref{conwf1}) and (\ref{conj}) we obtain that the form-factor (\ref{ratiof}),\vspace{-1ex}
\begin{align}
\nonumber
{\cal F} \big(\textbf v_{[N\backslash n]}, \textbf u_{N}, n\big)&=
\Bigg(\prod \limits_{l=n+1}^{N}v^{-2n}_l\Bigg)
\sum\limits_{\bla \subseteq \{{\CK}^{N-n}\}}S_{\bla} \big({\textbf v}^{-2}_{[N\backslash n]}\big)
S_{\widehat\bla} \big({\textbf u}_N^{2}\big)
\\
\label{field41}
&= \Bigg(\prod\limits_{l=n+1}^{N}v^{-2n}_l\Bigg)
\mathcal{P}_{\CK}\big(\textbf{v}_{[N\backslash n]}^{-2}, \textbf{u}_N^{2}\big) ,
\end{align}
is equal, up to pre-factor, to the generator of watermelons with deviation (\ref{gendev}).
Analogously,\vspace{-1ex}
\begin{gather}
{\cal F} (\textbf v_N, \textbf u_{N-n}, n) = \Bigg(\prod\limits_{l=1}^{N-n} u^{2n}_l \Bigg)
\mathcal{P}_{\CK} \big({\textbf v}^{-2}_N,{\textbf u}^2_{N-n}\big),
\label{fied41}
\end{gather}
where $\mathcal{P}_{\CK} \big({\textbf v}^{-2}_N,
{\textbf u}^2_{N-n}\big)$ corresponds to (\ref{gener1}).
Application of Theorem~\ref{theoremI} leads to the determinantal representations of the form-factors \eqref{field41} and \eqref{fied41}.

According to Theorem~\ref{theoremII},
the representation, for instance, (\ref{fied41}) taken in the $q$-para\-metri\-za\-tion (\ref{rep2111}) is the generating function of plane partitions in the box ${\cal B} (N-n, N, {\CK})$ (see \eqref{rep31}):\vspace{-1ex}
\begin{align*}
{\cal F} \bigl({\bf q}_{N}^{-\frac 12},({\bf q}_{N-n}/q)^{\frac 12}, n\bigr) & = q^{\frac{n}2(N-n)(N-n-1)}
{{\CP}_{\CK}} ({\bf q}_{N-n}/q, {\bf q}_N)
\\
& = q^{\frac
n2(N-n)(N-n-1)} Z_{q}^{{\rm}}(N-n, N,
{\CK}) .
\end{align*}
The number of plane partitions confined in this box (i.e., of watermelons with deviation) is gi\-ven~by\vspace{-1ex}
\begin{gather}
\lim\limits_{q\to 1} {\cal F} \bigl({\bf q}_{N}^{ -\frac 12},({\bf q}_{N-n}/q)^{\frac 12}, n\bigr) = A(N-n, N, {\CK}) .
\label{rep5}
\end{gather}

\subsection{Persistence of domain wall}

Let us turn to the persistence of domain wall (\ref{field0}):\vspace{-1ex}
\begin{gather}
\label{field0n}
{\cal F} \big({\bth}^{ \rm g}_{[N\backslash n]}, n, t\big) = {\rm e}^{-tN} \frac{\big\langle
\Psi \big({\bth}^{ \rm g}_{[N\backslash n]}\big) \big| \bar{\sf
F}^{+}_{n} {\rm e}^{- \frac t2 {\jum}} \bar{\sf F}_{n}
\big| \Psi \big({\bth}^{ \rm g}_{[N\backslash n]}\big)\big\rangle
}{\big\langle \Psi \big({\bth}^{ \rm g}_{[N\backslash n]}\big) \big|
\Psi \big({\bth}^{ \rm g}_{[N\backslash n]}\big)\big\rangle} ,
\end{gather}
where $N$ is the number of spins ``down'', ${\bth}^{ \rm g}_{[N\backslash n]}$ is $(N-n)$-tuple of solutions (\ref{grstxxn}), and ${\jum}$ is given by (\ref{hamr}).

The combinatorial properties we are interested in are expressed by the average in the numerator of (\ref{field0n}) taken under off-shell parametrization. Due to orthogonality of the states $| \bmu_N \rangle$ at fixed $N$, equations~(\ref{mel}) and (\ref{mel2})
lead us to the representation \cite{bmumn}:\vspace{-1ex}
\begin{gather}
\big\langle \Psi \big({\bf v}_{[N\backslash n]}\big) \big| \bar{\sf
F}_{n}^{+} {\rm e}^{- \frac t2 {\jum}} \bar{\sf F}_{n}
\big| \Psi ({\bf u}_{N-n})\big\rangle
\nonumber
\\ \qquad
{} =\Bigg(\prod\limits_{l=1}^{N-n}\frac{u^{2n}_l}{v^{2n}_{l+n}}\Bigg)
\sum\limits_{{\bla^{L, R}}\subseteq \{{\CK}^{N-n}\}}
S_{{\bla^L}}\big({\bf v}_{[N\backslash n]}^{-2}\big) S_{{\bla^R}}\big({\textbf u}_{N-n}^2\big) G \big({\widehat{\bmu}^L}; {\widehat{\bmu}^R} \big| t\big) , \label{rcl1}
\end{gather}
where
${\textbf u}^2_{N-n}$ and ${\textbf v}^2_{[N\backslash n]}$ stand for off-shell pa\-ra\-mete\-rization, and summations go over partitions~${\bla}^{R, L}$, $l({\bla}^{R, L}) = N-n$. The corresponding strict partitions ${\bmu}^{R, L} = {\bmu}^{R, L}_{N-n} = {\bla}^{R, L}_{N-n}+{\bdl}_{N-n}$ respect
\begin{gather}
\label{ratbe222}
M-n\ge \mu_1^{R, L}> \mu_2^{R, L} > \cdots >
\mu_{N-n}^{R, L} \ge 0 .
\end{gather}
The correlation function $G \big({\widehat{\bmu}^L}; {\widehat{\bmu}^R} \big| t\big)$ in (\ref{rcl1}) is defined by
(\ref{mpgf}),
\begin{gather*}
G \big({\widehat{\bmu}^L}; {\widehat{\bmu}^R} \big| t\big) =
\big\langle {\widehat{\bmu}^L} \big| {\rm e}^{- \frac t2 {\jum}} \big| {\widehat{\bmu}^R} \big\rangle ,
\end{gather*}
where $\big\langle {\widehat{\bmu}^L} \big|$ and $\big| {\widehat{\bmu}^R} \big\rangle$ are given by (\ref{ratbe2}). The partitions ${\widehat{\bmu}^{L, R}}$ are of length $l\big({\widehat{\bmu}^{L, R}}\big)=N$, and their parts
are related with those of ${\bmu}^{R, L}$, which respect (\ref{ratbe222}) (see (\ref{qanal14}) and (\ref{strhat})):
\begin{align*}
\big({\widehat\mu_1^{L, R}}, {\widehat\mu^{L, R}}_2, \dots, {\widehat\mu^{L, R}}_N\big)
& \equiv \big({\bmu}_{N-n}^{L, R} + {\bf n}_{N-n}, {\bdl}_n\big)
\\
& \equiv \big({\mu_1^{L, R}+n}, {\mu_2^{L, R}+n}, \dots, {\mu_{N-n}^{L, R}+n}, n-1, n-2, \dots, 1, 0\big). 
\end{align*}

Expanding the exponential in (\ref{rcl1}), we obtain:
\begin{gather*}
\big\langle \Psi \big({\bf v}_{[N\backslash n]}\big) \big| \bar{\sf
F}_{n}^{+} {\rm e}^{- \frac t2 {\jum}} \bar{\sf F}_{n}\big| \Psi ({\bf u}_{N-n})\big\rangle
= \sum_{K=0}^{\infty}\frac{(t/2)^K}{K!}
\big\langle \Psi \big({\bf v}_{[N\backslash n]}\big) \big| \bar{\sf
F}_{n}^{+}{(- {\jum})^K} \bar{\sf F}_{n}\big| \Psi ({\bf u}_{N-n})\big\rangle ,
\end{gather*}
where
\begin{gather}
 \big\langle \Psi \big({\bf v}_{[N\backslash n]}\big) \big| \bar{\sf
F}_{n}^{+}{(- {\jum})^K} \bar{\sf F}_{n}\big| \Psi ({\bf u}_{N-n})\big\rangle\nonumber
\\ \qquad
{} =\Bigg(\prod \limits_{l=1}^{N-n}\frac{u^{2n}_l}{v^{2n}_{l+n}} \Bigg)
\sum \limits_{{{\bla}^{L, R}}
\subseteq \{{\CK}^{N-n}\}}
S_{{{\bla}^L}}\big({\bf v}_{[N\backslash n]}^{-2}\big)
S_{{{\bla}^R}}\big({\textbf u}_{N-n}^2\big) \mathfrak{G}
\big({\widehat{\bmu}^L}; {\widehat{\bmu}^R} \big| K\big) .
\label{qanal020}
\end{gather}
The coefficient $\mathfrak{G}\big({\widehat{\bmu}^L}; {\widehat{\bmu}^R} \big| K\big)$ (see (\ref{qanal13})) gives
the number of nests of $K$-step paths between the sites labeled by $\widehat{\bmu}^L$ and $\widehat{\bmu}^R$:
\begin{gather}
\mathfrak{G}
\big({\widehat{\bmu}^L}; {\widehat{\bmu}^R} \big| K\big) =
\big\langle {\widehat{\bmu}^L} \big| (- {\jum})^K \big| {\widehat{\bmu}^R} \big\rangle .
\label{qanal20}
\end{gather}
By analogy with (\ref{qanal13}) and (\ref{qanal27}), we apply the commutation relations (\ref{comcom}) and (\ref{cr}) to right-hand side of (\ref{qanal20}) and obtain the number of nests of $K$-step trajectories
connecting the sites constituting ${\widehat\bmu^L}$ and ${\widehat\bmu^R}$:
\begin{gather*}
\big|P_K \big({\widehat{\bmu}^L} \rightarrow {\widehat{\bmu}^R}\big)\big|
= \mathfrak{G} \big({\widehat{\bmu}^L}; {\widehat{\bmu}^R} \big| K\big) .
\end{gather*}
The picture of transitions ${\widehat{\bmu}^L} \rightarrow {\widehat{\bmu}^R}$ is similar to that in
Figure~\ref{fig:f18}.

The transition amplitude $\big\langle \Psi \big({\bf v}_{[N\backslash n]}\big) \big| \bar{\sf F}_{n}^{+}{(- {\jum})^K} \bar{\sf F}_{n} \big| \Psi \big({\bf u}_{N-n})\big\rangle$ (\ref{qanal020}) is interpreted in terms of nests of self-avoiding lattice paths made by $N$ vicious walkers. The walkers initially occupying the points $C_i$ move by the lock steps rules to the sites $\widehat{\bmu}^L$ (see Definition~\ref{Definition3b} and equation~\eqref{schdevi}). The transition between $\widehat{\bmu}^L$ and
$\widehat{\bmu}^R$
occurs in $K$ steps in accordance with the random turns model (Section~\ref{randfer}). After this, they move from $\widehat{\bmu}^R$ to points $B_i$ by the lock steps rules (see Definition~\ref{Definition4b} and equation~\eqref{lim}).
A typical nest is presented in Figure \ref{fig:f9}.
\begin{figure}[h]
\center
\includegraphics [scale=.9] {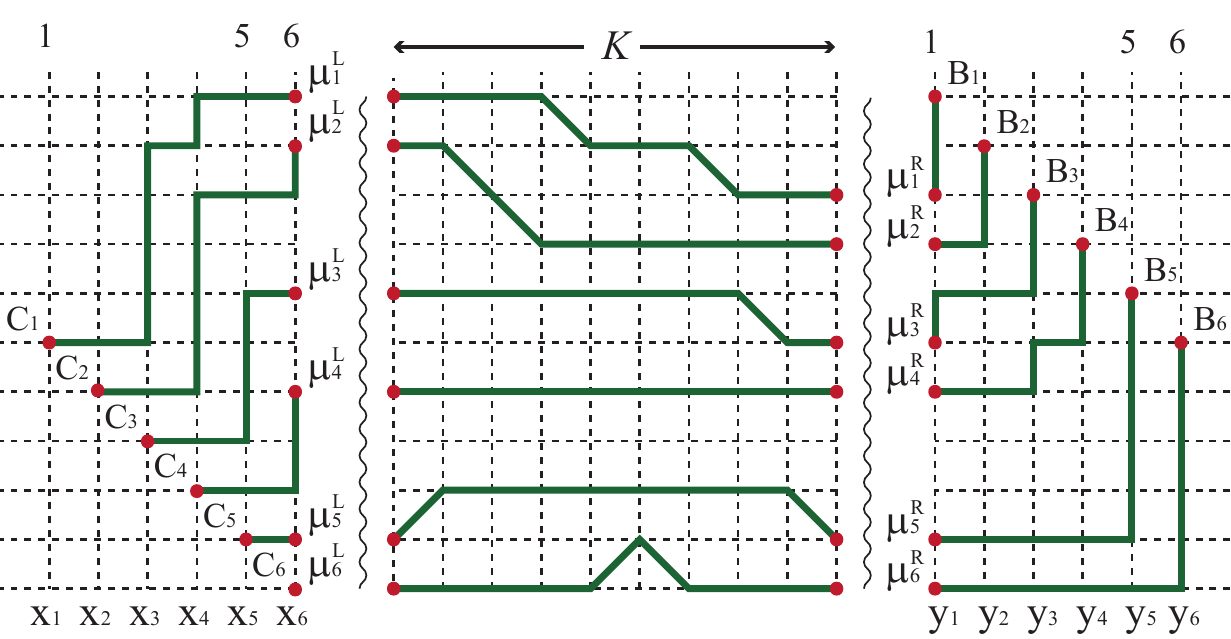}
\caption{Nest of paths contributing to
$\big\langle \Psi \big({\bf v}_{[N\backslash n]}\big) \big| \bar{\sf
F}_{n}^{+}{(- {\jum})^K} \bar{\sf F}_{n}
\big| \Psi ({\bf u}_{N-n})\big\rangle$ at $N=6$, $n=2$ and $K=9$ (the gluing is along the wavy lines, where ${\widehat\bmu^{L, R}}$ are shown without the hats).}
\label{fig:f9}
\end{figure}

Equation which governs $G \big({\widehat\bmu^L}; {\widehat\bmu^R} | t\big)$ (\ref{ratbe222}) is the same as (\ref{phem}), and its solution is given either by \eqref{det} or by (\ref{ratbe7}),
\begin{gather}
 G \big({\widehat\bmu^L}; {\widehat\bmu^R} \big| t\big)
 = \frac{{\rm e}^{t N}}{(M+1)^N}\sum\limits_{\{{\bphi}_N\}}
{\rm e}^{- t E_N({\bphi}_N)}\big|{\CV} \big({\rm e}^{{\rm i} {\bphi}_N}\big)\big|^2
S_{{\widehat\bla}^L}\big({\rm e}^{{\rm i} {\bphi}_N}\big)
S_{{\widehat\bla^R}}\big({\rm e}^{-{\rm i} {\bphi}_N}\big) ,
 \label{qanal24}
 \end{gather}
where ${\widehat\bla}^{L, R} =
\widehat \bmu^{L, R} - \bdl_N$.
The average (\ref{rcl1}) is re-expressed with regard at \eqref{qanal24}:
\begin{align}
\big\langle \Psi \big({\bf v}_{[N\backslash n]}\big) \big| \bar{\sf F}_{n}^{+} {\rm e}^{- \frac t2 {\jum}} \bar{\sf F}_{n}\big| \Psi ({\bf u}_{N-n})\big\rangle ={}&
\Bigg(\prod\limits_{l=1}^{N-n}
\frac{u^{2n}_l}{v^{2n}_{l+n}}\Bigg) \frac{{\rm e}^{t N}}{(M+1)^N}
\sum\limits_{\{{\bphi}_N\}}
{\rm e}^{- t E_N({\bphi}_N)} \big|{\CV} \big({\rm e}^{{\rm i} {\bphi}_N}\big)\big|^2\nonumber
\\
&\times {\cal P}_{\cal M}\big({\bf v}_{[N\backslash n]}^{-2}, {\rm e}^{{\rm i} {\bphi}_N}\big)
{\cal P}_{\cal M} \big({\rm e}^{-{\rm i} {\bphi}_N}, {\textbf u}_{N-n}^2\big), \label{field81}
\end{align}
where ${\cal P}_{\cal M}\big({\bf v}_{[N\backslash n]}^{-2}, {\rm e}^{{\rm i} {\bphi}_N}\big)$ and ${\cal P}_{\cal M} \big({\rm e}^{-{\rm i} {\bphi}_N}, {\textbf u}_{N-n}^2\big)$ are the generators of watermelons with deviation~$n$ parameterized by the numbers of steps according to (\ref{gendev}) and (\ref{gendev2}). This representation is useful in studying the asymptotical behaviour of \eqref{field0n}.

Eventually, the representations (\ref{norma}), (\ref{mpgf}), and (\ref{field81}) allow us to express the persistence of domain wall (\ref{field0}) (or \eqref{field0n}) parameterized by \eqref{grstxxn}:
\begin{align}
{\cal F}\big({\bth}_{[N \backslash n]}^{ \rm g}, n, t \big)
={}& \frac{\big|\CV \big({\rm e}^{{\rm i}{\bth}_{[N\backslash n]}^g}\big)\big|^2}{ (M+1)^{2N}}\sum\limits_{\{{\bphi}_N\}}{\rm e}^{- t E_N({\bphi}_N)}\nonumber
\\
&\times
 \big|{\CV} ({\rm e}^{{\rm i} {\bphi}_N})\big|^2 {\cal P}_{\cal M}\big({\rm e}^{- i {\bth}_{[N\backslash n]}^{ \rm g}}, {\rm e}^{{\rm i} {\bphi}_N}\big){\cal P}_{\cal M} \big({\rm e}^{-{\rm i} {\bphi}_N}, {\rm e}^{{\rm i} {\bth}_{[N\backslash n]}^{ \rm g}}\big) .
\label{field181}
\end{align}

\subsection{Auto-correlation function}

Let us turn to the auto-correlation function (\ref{auto}). Analogously to \eqref{rcl1}, we obtain the average under off-shell parametrization:
\begin{gather}
 \big\langle \Psi \big({\bf v}_{[N\backslash n]}\big) \big| \bar{\sf
F}_{n}^{+} {\rm e}^{- \frac{t_1}2 {\jum}} {\bar\varPi}_m {\rm e}^{- \frac{t_2}2 {\jum}} \bar{\sf F}_{n} \big| \Psi ({\bf u}_{N - n})\big\rangle \nonumber
\\ \qquad
{} =
\Bigg(\prod\limits_{l=1}^{N-n}\frac{u^{2n}_{l}}{v^{2n}_{l+n}}\Bigg)
\sum\limits_{{\bla^{L, R}}\subseteq \{{\CK}^{N-n}\}}
S_{{\widehat{\bla}^L}}\big({\bf v}_{[N\backslash n]}^{-2}\big) S_{{\widehat{\bla}^R}}\big({\textbf u}_{N - n}^2\big) G \big({\widehat{\bmu}^L}; {\widehat{\bmu}^R}\big| t_1, t_2, m\big) ,
\label{twop13}
\end{gather}
where summation is similar to that in \eqref{rcl1}, and ${\textbf u}^2_{N - n}$, ${\textbf v}^2_{[N\backslash n]}$ stand for off-shell par\-am\-etr\-ization (see \eqref{xindi1}). The generating
function $G \big({\widehat{\bmu}^L}; {\widehat{\bmu}^R} \big| t_1, t_2, m\big)$ is given by (\ref{twop7}):
\begin{gather}
G \big({\widehat{\bmu}^L}; {\widehat{\bmu}^R} \big| t_1, t_2, m\big) =
\big\langle {\widehat{\bmu}^L} \big| {\rm e}^{ - \frac{t_1}2 {\jum}} {\bar\varPi}_m {\rm e}^{- \frac{t_2}2 {\jum}} \big| {\widehat{\bmu}^R} \big\rangle . \label{twop14}
\end{gather}

Expanding (\ref{twop14}) with respect of $\bigl(\frac{t_1}2\bigr)^{K_1}$ and $\bigl(\frac{t_2}2\bigr)^{K_2}$, we obtain a double series characterized by the coefficients (\ref{twop19}) as follows:
\begin{gather}
{\mathfrak G} \big({\widehat{\bmu}^L};{\widehat{\bmu}^R} \big| K_1, K_2, m\big) =
\big\langle {\widehat{\bmu}^L} \big| (-{\jum})^{K_1} {\bar\varPi}_m (-{\jum})^{K_2} \big| {\widehat{\bmu}^R} \big\rangle . \label{twop15}
\end{gather}
Equation \eqref{twop21}, as well as its solution \eqref{twop111}, are valid for the coefficients (\ref{twop15}), which enumerate the nests of $(K_1 + K_2)$-step random turns walks
connecting the sites ${\widehat{\bmu}^L}$ and ${\widehat{\bmu}^R}$. Equation~(\ref{twop112}) is also modified and takes the form:
\begin{align}
\nonumber
\big|P_{K_1 + K_2} \big({\widehat{\bmu}^L} \underset{m}\rightarrow {\widehat{\bmu}^R}\big)\big|
& = {\mathfrak G} \big({\widehat{\bmu}^L}; {\widehat{\bmu}^R} \big| K_1, K_2, m\big)
\\
& = \sum_{\bro \subseteq\{({\cal M}\slash m)^N\}}
\big|P_{K_1} \big({\widehat{\bmu}^L} \rightarrow {\bro}+{\bdl}_N \big) \big|
\big|P_{K_2} \big({\bro}+{\bdl}_N \rightarrow {\widehat{\bmu}^R} \big)\big|.
\label{twop17}
\end{align}
Equation (\ref{twop17}) enumerates self-avoiding lattice paths of random turns vicious walkers going from ${\widehat{\bmu}^L}$ to ${\widehat{\bmu}^R}$ and infiltrating
through a bottleneck: the projector forbids ``visiting'' the sites from $0^{\rm th}$ to $(m-1)^{\rm th}$. The correlation function (\ref{twop13}) corresponds to summation over the nests with all possible $K_1$ and $K_2$ and over all ${{\bmu}^L}$ and ${{\bmu}^R}$.

The transition amplitude (\ref{twop13}) admits a representation of the same type as \eqref{field81}, which includes the generating functions of watermelons ${\cal P}_{{\cal M}/m}$ (\ref{gener1}) and the generating functions of watermelons with deviation ${\cal P}_{\cal M}$, (\ref{gendev}) and (\ref{gendev2}). With regard at
\eqref{twop7} and \eqref{twop13},
\eqref{twop14}, the auto-correlation function
(\ref{auto}) takes the form:
\begin{gather*}
 \Gamma\big({\bth}_{[N\backslash n]}^{ \rm g}, N, n, m, t_1, t_2\big)
 \\ \qquad
{}= \frac{\big|\CV \big({\rm e}^{{\rm i}{\bth}_{[N\backslash n]}^g}\big)\big|^2}{(M+1)^{N}}
\sum\limits_{{\bla^{L, R}}\subseteq \{{\CK}^{N-n}\}}
S_{{\widehat{\bla}^L}}\big({\rm e}^{-{\rm i} {\bth}_{[N\backslash n]}^g}\big) S_{{\widehat{\bla}^R}}\big({\rm e}^{{\rm i} {\bth}_{[N\backslash n]}^g}\big) {\cal G} \big({\widehat{\bmu}^L}; {\widehat{\bmu}^R} \big| t_1, t_2, m\big) , 
\end{gather*}
or, equivalently,
\begin{gather}
 \Gamma({\bth}_{[N\backslash n]}^{ \rm g}, N, n, m, t_1, t_2)\nonumber
 \\ \quad
{}= \frac{|\CV ({\rm e}^{{\rm i}{\bth}_{[N\backslash n]}^g})|^2}{(M+1)^{3 N}}
\!\sum\limits_{\{{\bphi}_N\}}\!\sum\limits_{\{{\bphi}^{\prime}_N\}}\!
{\rm e}^{-t_1 E_N ({\bphi}^{\prime}_N)- t_2 E_N ({\bphi}_N)}
\big|{\CV} \big({\rm e}^{{\rm i} {\bphi}_N}\big)\big|^2 \big|{\CV} \big({\rm e}^{{\rm i} {\bphi}^{\prime}_N}\big)\big|^2
{\cal P}_{\cal M}\big( {\rm e}^{-{\rm i} {\bth}_{[N\backslash n]}^g}, {\rm e}^{{\rm i} {\bphi}_N}\big)\nonumber
\\ \quad\hphantom{= \frac{|\CV ({\rm e}^{{\rm i}{\bth}_{[N\backslash n]}^g})|^2}{(M+1)^{3 N}}
\!\sum\limits_{\{{\bphi}_N\}}\!\sum\limits_{\{{\bphi}^{\prime}_N\}}\!}
{} \times
{\cal P}_{{\cal M}/m} \big({\rm e}^{-{\rm i} {\bphi}_N}, {\rm e}^{{\rm i} {\bphi}^{\prime}_N}\big) {\cal P}_{\cal M} \big({\rm e}^{-{\rm i} {\bphi}^{\prime}_N}, {\rm e}^{{\rm i} {\bth}_{[N\backslash n]}^g}\big) .
\label{auansw}
\end{gather}
The representation \eqref{auansw}
is reduced at $m=0$ to \eqref{field181}.

\section{Asymptotics of dynamical correlation functions}\label{asymp}

The representations (\ref{field181}) and (\ref{auansw}) are related with the corresponding generating functions of watermelons given by
equations~\eqref{volwater3}, \eqref{gfw5} (Proposition~\ref{Proposition3}) and \eqref{volwaterk3}, \eqref{volwaterk31} (Proposition~\ref{Proposition4}). In turn, the asymptotics of these representations are related with the numbers of plane partitions~(\ref{rep55}) in ${\cal B}(N, N, {\CK}-m)$
and the numbers of plane partitions~(\ref{rep5}) in ${\cal B}(N-n, N, {\CK})$.

Indeed, let us obtain the
leading estimates for the correlation functions at large evolution parameter $t\gg 1$ in the case of long enough chain, $M \gg 1$, and moderate $N$, $1\ll N \ll M$ (double scaling limit). The discrete parameters of summation in (\ref{field181}) and (\ref{auansw}), the same as in \eqref{ratbe7}, fill at large enough $M$ the segment $[ -\pi, \pi ]\ni\phi$. The corresponding summations are replaced approximately by multiple integration, and the matrix integrals arise for (\ref{field181}) and~(\ref{auansw}). The leading approximation for the integrals is provided by steepest descent \cite{din} due to $\cos\phi\simeq 1$ at $t\gg 1$.

Firstly, the transition amplitude \eqref{mpcf} given by \eqref{mpgf} and by the integral expressing
\eqref{qanal24} behaves at $t\gg 1$ in leading approximation:
\begin{gather}\label{aszero}
{\cal G} \big(\widehat{\bmu}^L; \widehat{\bmu}^R \big| t\big)\simeq S_{\widehat{\bla}^L}({\textbf 1}) S_{{\widehat{\bla}^R}} ({\textbf 1}) \frac{{\cal I}_N}{t^{{N^2}/{2}}} ,
\end{gather}
where $S_{\widehat{\bla}^{L, R}} ({\textbf 1})$ are given by (\ref{numbpaths}), and ${\cal I}_N$ is a special form of {\it Mehta's integral} connected with the partition function of Gaussian unitary ensemble \cite{meh}:
\begin{gather}
\displaystyle{{\cal I}_N \equiv
\frac{1}{N!}\int \limits_{-\infty}^{\infty}\!\int\limits_{-\infty}^{\infty}
\cdots\int\limits_{-\infty}^{\infty} {\rm e}^{-\frac{1}2\sum\limits_{l=1}^N
x^2_l}\!\!\prod_{1\leq k<l\leq N} \bigl|x_k-
x_l\bigr|^2 \prod_{k=1}^{N} \frac{{\rm d} x_k}{2\pi} } = \prod_{l=0}^{N-1} \frac{l!}{(2\pi)^{1/2}} .
\label{spdfpxx60}
\end{gather}
Provided that right-hand side of \eqref{spdfpxx60} is expressed through the Barnes $G$-function \cite{barn},
the following estimate arises \cite{bmnph}:
\begin{gather}
\log {\cal I}_N = \frac{N^2}{2} \log N -\frac{3 N^2}{4} + {\cal O}(\log N) ,\qquad N\gg 1 . \label{spdfpxx6}
\end{gather}

Analogously to \eqref{aszero}, the persistence of domain wall ${\cal F} \big({\bth}_{[N\backslash n]}^{ \rm g},n, t\big)$ (\ref{field181}) behaves at $t\gg 1$:
\begin{gather}\label{asp}
{\cal F} \big({\bth}_{[N\backslash n]}^{ \rm g},n, t\big) \simeq \frac{\CA(N, n)}{t^{{N^2}/{2}}} ,\\
\label{asp1}
{\CA}(N, n) \equiv A^2 (N-n, N, M-N+1)
\bigg(\frac{2\pi}{M+1}\bigg)^{N^2} {\cal I}_N^3 ,
\end{gather}
where the estimate
\begin{gather*}
\frac{1}{{\CN}^2 \big({\bth}^{ \rm g}_{[N\backslash n]}\big)}
\simeq \frac{(2\pi)^{(N-n)(N-n-1)}}{(M+1)^{(N-n)^2}} \prod_{1\leq r<s\leq N-n}
|r-s|^2 \underset{N\gg n}\approx\bigg(\frac{2\pi}{M+1}\bigg)^{N^2} {\cal I}_N^{2}
\end{gather*}
is taken into account \cite{bmumn}.

The coefficient ${\CA}(N, n)$ \eqref{asp1} includes $A^2(N-n, N, M-N+1)$, where
\begin{gather}
A(N-n, N, M-N+1) = \lim_{q\to 1} W_{q} (N-n, N, {\CK}) ,
\label{spdfpxx61}
\end{gather}
and $W_{q} (N-n, N, {\CK})$
is the generating function of watermelons with deviation \eqref{volwaterk3} or \eqref{volwaterk31}.
Equivalently, $A(N-n, N, M-N+1)$
is the number of plane partitions (\ref{wmdbd11}) in a box \mbox{${\mathcal{B}}(N-n, N, M-N+1)$} of the height $M-N+1$ and $(N-n)\times N$ bottom.

The asymptotics (\ref{aszero}) of the transition amplitude (\ref{mpcf}) and the asymptotics (\ref{asp}) of the domain wall persistence (\ref{field0})
demonstrate the power law decay governed by the same critical exponent $N^2/2$, however different combinatorial factors
characterize the amplitudes.

Combining \eqref{twop7},
\eqref{twop11} and \eqref{qanal24}, we estimate the transition amplitude
${\cal G} \big({\widehat{\bmu}^L}; {\widehat{\bmu}^R} \big| t_1, t_2, m\big)$
(see also \eqref{twop777}):
\begin{gather}
{\cal G} \big({\widehat{\bmu}^L}; {\widehat{\bmu}^R} \big| t_1, t_2, m\big) \simeq
S_{\widehat{\bla}^L} ({\textbf 1}) S_{\widehat{\bla}^R} ({\textbf 1})
 A(N, N, M-m-N+1)\frac{{\cal I}^2_N}{(t_1 t_2)^{{N^2}/{2}}}.
\label{autoas4}
\end{gather}
The coefficient $A(N, N, M-m-N+1)$ is given:
\begin{gather*}
A(N, N, M-m-N+1) = \lim_{q\to 1} W_{q} (N, {\CK} / m) ,
\end{gather*}
where $W_{q} (N, {\CK} / m)$
is the generating function of the watermelons without deviation \eqref{volwater3} or~\eqref{gfw5}. Equivalently, $A(N, N, M-m-N+1)$ is the number \eqref{npp} of plane partitions in the box $\mathcal{B} (N, N, M-m-N+1)$ of the height $M-m-N+1$ with square $N\times N$ bottom.

Eventually, the representation (\ref{auansw}) allows us to estimate the auto-correlation function (\ref{auto})
when $1\ll N \ll M$ and $t_1, t_2 \gg 1$:
\begin{gather}
\label{autoas1}
 \Gamma\big({\bth}_{[N\backslash n]}^{ \rm g}, N, n, m, t_1, t_2\big) \simeq
 \frac{\CA(N, n, m)}{(t_1 t_2)^{{N^2}/{2}}},
\\
\label{autoas2}
{\CA}(N, n, m) \equiv A(N, N, M-m-N+1)A^2 (N-n, N, M-N+1)
\bigg(\frac{2\pi}{M+1}\bigg)^{N^2} {\cal I}_N^4 .
\end{gather}
The critical exponent is the same in \eqref{autoas4} and \eqref{autoas1}, whereas the combinatorial amplitudes are different.

The representation of the generating function $W_{q} (N-n, N, {\CK})$ of the watermelons with deviation is due to Theorem~\ref{TheoremIV} provided that $\dl=n$ and $M\to\infty$. Therefore it is appropriate to use \eqref{qdetgv370} in the limiting formula \eqref{spdfpxx61} for studying \eqref{asp} or \eqref{autoas1} for the chain of infinite length.
Equation \eqref{qdetgv370}, as the norm-trace generating function, enables enumeration of the plane partitions with fixed traces of diagonal entries \cite{statm}.

We have shown that both asymptotics, \eqref{asp} or \eqref{autoas1}, are related with enumeration of the random walks of the lock step type. The asymptotical behavior of the combinatorial contributions $A(N, N, M-m-N+1)$
and $A (N-n, N, M-N+1)$ in the double scaling limit is provided
by \cite{bmumn}. Together with \eqref{spdfpxx6}, the limiting behavior of \eqref{asp} and \eqref{autoas1} is completely described in the double scaling limit.

\section{Conclusion}\label{conc}

In this paper we have dealt with the calculation of the auto-correlation functions of
the $XX0$ Heisenberg spin chain. We have presented the exact expressions of the correlation functions in terms of different types of nests of
self-avoiding lattice walks, namely in terms of lock steps and random turns walks. In particular, the representation of the persistence of domain wall correlation function admits a visualization in terms of the lattice paths infiltrating through a bottleneck.
Such a description is an essential result of our study.
The asymptotical estimates
in the double scaling limit at large value of the evolution parameter demonstrate
that the correlation functions in question are multiples of the numbers of the watermelons
or, equivalently, of the numbers of the plane partitions distributed in a high box.

\subsection*{Acknowledgments}
This work was supported by the Russian Science Foundation (Grant 18-11-00297).
We are grateful to the referees for their comments and suggestions which enabled to improve the paper.

\pdfbookmark[1]{References}{ref}
\LastPageEnding

\end{document}